\documentclass[11pt]{article}
\usepackage{units}
\usepackage{xcolor}
\definecolor{Orchid}{rgb}{0.3568627450980392, 0.21176470588235294,0.5568627450980392}
\definecolor{Peach}{rgb}{0.9137254901960784, 0.5725490196078431, 0.17254901960784313}
\usepackage{appendix}
\usepackage{subfig}
\usepackage{caption}
\usepackage{tikz}
\usetikzlibrary{shapes.geometric, arrows.meta, positioning}

\tikzset{
    block/.style={
        rectangle,
        draw,
        text width=5em,
        text centered,
        minimum height=12mm,
        node distance=2cm, 
    },
    state/.style={
        rectangle,
        draw,
        minimum width=5em,
        minimum height=12mm,
        text centered
    }
}
\usepackage{graphicx}
\usepackage{pgfplots}
\usetikzlibrary{matrix,arrows,decorations.pathmorphing}
\usepackage{pgfplotstable}
\usepgfplotslibrary{fillbetween}
\pgfplotsset{compat=1.8, plot coordinates/math parser=false}
\usepackage[percent]{overpic}
\usepackage{booktabs}
\usetikzlibrary{plotmarks}
\DeclareUnicodeCharacter{2212}{−}
\usepgfplotslibrary{groupplots,dateplot}
\usetikzlibrary{patterns,shapes.arrows}
\usetikzlibrary{positioning, arrows.meta}
\usepackage{amsmath,amssymb,amscd,amsthm,ifthen,color,framed,bm}
\allowdisplaybreaks
\usepackage{algorithm}
\usepackage{algpseudocode}
\algnewcommand{\Initialize}[1]{%
  \State \textbf{Initialize:}
  \Statex \hspace*{\algorithmicindent}\parbox[t]{.8\linewidth}{\raggedright #1}
}
\algnewcommand{\Req}[1]{%
  \State \textbf{Require:}
  \Statex \hspace*{\algorithmicindent}\parbox[t]{.8\linewidth}{\raggedright #1}
}
\usepgfplotslibrary{fillbetween}
\usetikzlibrary{positioning, arrows.meta}
\usepackage{rh_defs_21}
\usetikzlibrary{decorations.pathreplacing}
\usetikzlibrary{decorations.pathreplacing,calligraphy}
\usepgfplotslibrary{fillbetween}

\usepackage[
    backend=biber,
    url=false,
    isbn=false,
    doi=false,
    date=year,
    hyperref=auto,
    style=alphabetic-verb,
    natbib=true,
    sorting=nty,
    giveninits=true,
    maxnames=7,
    maxbibnames=10,
]{biblatex}
\DeclareSourcemap{
    \maps[datatype=bibtex]{
        \map{
            \step[fieldset=language, null]
            \step[fieldset=extra, null]
            \step[fieldset=volume, null]
            \step[fieldset=number, null]
            \step[fieldset=pages, null]
            \step[fieldset=note, null]
            \step[fieldset=series, null]
            \step[fieldset=publisher, null]
        }
    }
}
\bibliography{constrained_coding.bib} 

\usepackage{comment}

\usepackage{hyperref}
\hypersetup{
    unicode=false,          
    pdftoolbar=true,        
    pdfmenubar=true,        
    pdffitwindow=false,     
    pdfstartview={FitH},   
    pdfauthor={Reinhard Heckel},    
    pdfsubject={Subject},   
    pdfproducer={Producer},
    pdfnewwindow=true,      
    colorlinks=true,     
    linkcolor=red,          
    citecolor=green,        
    filecolor=magenta,      
    urlcolor=cyan           
}

\newtheorem{lemma}{Lemma}
\newtheorem{definition}{Definition}

\newtheorem{theorem}{Theorem}

\usepackage{bbm}

\newcounter{dummy}
\begin{document}
\begin{center}

{\bf{\LARGE{
Embracing Errors Is More Efficient Than Avoiding Them Through Constrained Coding for DNA Data Storage
}}}
\stepcounter{dummy} 
\vspace*{.2in}

{\large{
\begin{tabular}{cccc}
Franziska Weindel $^{\ast}$, Andreas L. Gimpel $^{\ddagger} $, Robert N. Grass $^{\ddagger}$, and Reinhard Heckel $^{\ast}$
\end{tabular}
}}

\vspace*{.05in}

\begin{tabular}{c}
$^\ast$Dept. of Electrical and Computer Engineering, Technical University of Munich, \\ Arcistrasse 21, 
 80333, Munich, Germany \\
$^{\ddagger} $ Dept. of Chemistry and Applied Biosciences, ETH Zürich, \\
Vladimir-Prelog-Weg 1-5, 8093, Zurich, Switzerland
\end{tabular}

\vspace*{.1in}

\today

\vspace*{.1in}

\end{center}

\markboth{IEEE TRANSACTIONS ON INFORMATION THEORY}%
{Shell \MakeLowercase{\textit{et al.}}: Bare Demo of IEEEtran.cls for Journals}


\begin{abstract}
DNA is an attractive medium for digital data storage. When data is stored on DNA, errors occur, which makes error-correcting coding techniques critical for reliable DNA data storage. To reduce the errors, a common technique is to include constraints that avoid homopolymers (consecutive repeated nucleotides) and balance the GC content, as sequences with homopolymers and unbalanced GC content are often associated with higher error rates. However, constrained coding comes at the cost of an increase in redundancy. An alternative is to control errors by randomizing the sequences, embracing errors, and paying for them with additional coding redundancy. In this paper, we determine the error regimes in which embracing substitutions is more efficient than constrained coding for DNA data storage. 
Our results suggest that constrained coding for substitution errors is inefficient for existing DNA data storage systems. Theoretical analysis indicates that for constrained coding to be efficient, the increase in substitution errors for nucleotides in homopolymers and sequences with unbalanced GC content must be very large. Additionally, empirical results show that the increase in substitution, deletion, and insertion rates for these nucleotides is minimal in existing DNA storage systems.

\end{abstract}

\section{Introduction}

DNA data storage is an emerging storage medium due to its high density, longevity, and energy efficiency. In DNA data storage, a string of bits is converted into multiple DNA sequences composed of the four bases adenine (A), cytosine (C), guanine (G), and thymine (T). 
The DNA sequences can be stored for long periods and read using sequencing technologies. 

However, synthesis (writing), storage, and sequencing (reading) are error-prone. Thus, reliable data storage can only be achieved using error-correcting codes. Error-correcting codes allow for error detection and correction at the cost of added redundancy. Since synthesis and sequencing are expensive, the goal is to achieve reliable data storage using minimal redundancy, i.e., design schemes that maximize the code rate (ratio of the number of information bits to the total number of nucleotides synthesized), while achieving a vanishing probability of decoding error as the sequence length increases.

The code rate is upper-bounded by the channel capacity and the optimal level of redundancy depends on the error rates of the DNA storage system. 
A common technique to reduce the number of errors is constrained coding, used by early works on DNA data storage~\citep{goldman_towards_2013,grass_robust_2015,bornholt_dna-based_2016, church_next-generation_2012}.  
Constrained coding avoids systematic errors by removing error-prone sequences from the set of possible sequences. 

In DNA data storage, systematic errors are related to the biochemical structure of the DNA sequences, and experiments have shown that sequences with homopolymers (consecutive repeated nucleotides) and unbalanced GC content have higher error rates~\citep{ross_characterizing_2013,bragg_shining_2013,bohlin_estimation_2019,stoler_sequencing_2021}. 
Therefore, much research is devoted to code constructions that constrain the length of homopolymers and balance the GC content to improve the reliability of DNA storage systems~\citep{immink_properties_2020, dube_dna_2019,benerjee_homopolymers_2022,erlich_dna_2017,nguyen_capacity-approaching_2021,press_hedges_2020,park_iterative_2022}. 

\begin{figure}[t]
\centering
    \includegraphics[scale=0.27]{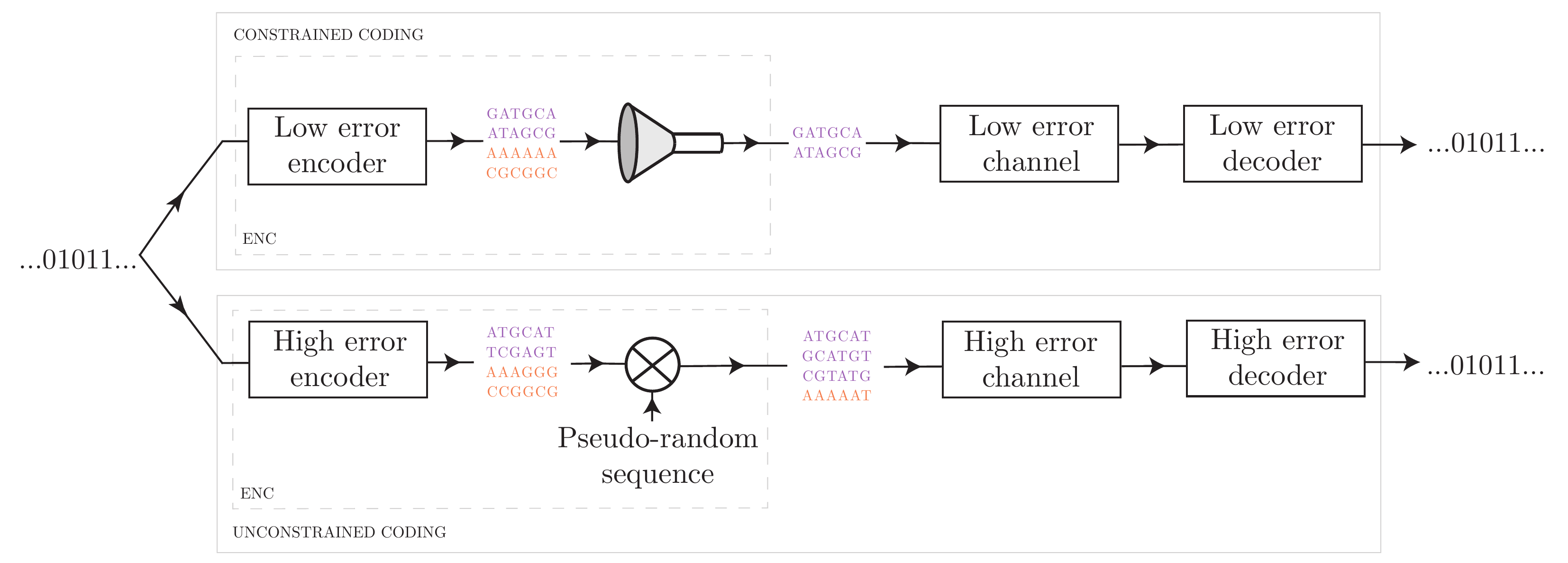}
\captionsetup{labelfont=bf}
\caption{
Constrained coding removes error-prone sequences to reduce the number of errors at the cost of fewer sequences available to store information. Unconstrained coding controls the error rate by modulo-4 addition of the input sequences with a pseudo-random sequence. This reduces the occurrence of error-prone sequences, but may require more coding redundancy to achieve a vanishing probability of decoding error as the sequence length increases.
}
\label{fig1}
\end{figure}

However, limiting the number of sequences available to store information reduces the code rate, as already discussed in \citet{shannon_mathematical_1948}'s seminal work on information theory.
Consequently, there is a trade-off between maximizing the number of sequences for information storage and improving system reliability.
In this paper, we explore this trade-off by comparing two approaches to code design, as illustrated in Figure~\ref{fig1}: 
\begin{itemize}
\item \textbf{Constrained coding.} Constrained coding excludes certain sequences (e.g., those containing homopolymers). This has the potential to reduce errors, but at the cost of fewer sequences available to store information. Constrained coding has been adapted in early DNA storage systems~\citep{goldman_towards_2013,grass_robust_2015,bornholt_dna-based_2016, church_next-generation_2012}. 

\item \textbf{Unconstrained coding.} 
An alternative is to embrace errors and not exclude any sequences, but to minimize structural errors through randomization. Randomization ensures that sequences are random for statistical purposes; thus long homopolymers occur with low probability and the mean GC content is balanced. However, 
allowing all possible sequences for information storage may come at the cost of higher error rates, which must be controlled with more coding redundancy. This approach has been adapted in later DNA storage systems~\citep{antkowiak_low_2020,organick_random_2018}. 
\end{itemize}

Our goal is to understand whether constrained or unconstrained coding maximizes the code rate. We study constrained coding for homopolymers and GC content in two distinct settings, focusing only on substitution errors.  

Our main theoretical finding is that for constrained coding to be efficient, the increase in substitution errors for nucleotides in homopolymers and sequences with unbalanced GC content must be very large. Additionally, the empirical results show that the increase in substitution, deletion, and insertion rates for these nucleotides is minimal in existing DNA storage systems.

The paper is organized as follows. In Section~\ref{Homopolymers}, we analyze constraints on homopolymer length and provide achievable code rates for both coding schemes in the presence of substitution errors. We then use these results to determine the efficiency of constrained versus unconstrained coding in terms of achievable code rates. In Section~\ref{GC-Content}, we study constraints on GC content and state Gilbert-Varshamov based lower bounds for the code rate for both coding schemes in the presence of substitution errors. We then use these results to determine the error regimes in which constrained coding for GC content achieves a higher lower bound than unconstrained coding. Finally, in Section~\ref{emp}, we present experimental results on substitution, deletion, and insertion rates as a function of homopolymer length and GC content to determine the error regimes in which existing DNA storage systems lie.

\section{Related Work}\label{sec2}

There is a large body of work on homopolymer and GC content constrained codes for DNA data storage, motivated by two factors. 
First, long stretches of homopolymers and sequences with unbalanced GC content are challenging to read and write. 
Second, run-length limited and direct current (DC) free codes are very successful and prevalent in conventional data storage. Constrained coding is widely used for data storage on hard disks, optical media, and magnetic tapes (see \cite{schouhamer_immink_codes_1999} Table 1.1 for an overview), which has sparked interest in its possible application to DNA data storage.

Run-length limited codes, which are common in conventional data storage, are similar to homopolymer constrained codes. However, their redundancy can be offset by a gain in data density. Run-length limited or $(d,k)$-codes generate codewords with a minimum of $d$ and a maximum of $k$ binary zeros between binary ones. Hence, homopolymer constrained codes can be regarded as a subclass of $(d,k)$-codes. In magnetic and optical recording, a minimum run-length constraint is imposed to reduce the inter-symbol interference between adjacent transmissions. The goal is to increase system reliability, similar to how homopolymer and GC content constraints are designed to reduce the average number of errors. However, in conventional data storage, the code rate loss due to the minimum run-length constraint can be compensated for by increasing the clock rate. 
The minimum run-length constraint ensures sufficient time between adjacent transmissions, allowing for shorter bit windows and more data to be stored on the same physical space \citep{simic_coding_1995,zhang_performance_2007}. In DNA data storage, the loss in code rate due to homopolymer and GC content constraints cannot be offset by a gain in data density. Therefore, it is not clear whether the success of run-length limited codes in conventional data storage translates to DNA data storage.

Similarly, DC free codes, which are common in conventional data storage, share parallels in code construction with GC content constrained codes, although their objectives differ. 
DC free codes balance the number of binary zeros and ones. However, while GC content constrained codes are designed to reduce the average number of errors, DC free codes are designed to maintain the synchronization of the decoder \citep{schouhamer_immink_codes_1999}. For example, in optical discs, DC free codes prevent written data from interfering with the servo system. The efficiency of these codes cannot be directly applied to DNA data storage, because error detection and correction are challenging when the synchronization of the decoder is lost \citep{immink_efmplus_1995}.

The benefits of constrained coding may be compromised when synchronizing in a noisy channel and in systems in which errors are not limited to specific error-prone bit patterns. For example, \citet{kautz_fibonacci_1965} considers synchronization in a noisy channel. He suggests adding a minimum run-length constraint~$d$ to reduce the probability that the maximum run-length constraint~$k$ is violated due to noise. However, \citet{tang_block_1970} show that for synchronization in a noisy channel, $(d,k)$~-codes lead to strictly lower code rates than, for example, using a fraction of the maximum run-length~$k$. 
\citet{immink_weakly_1997} introduces the concept of weakly constrained codes and shows that they can achieve higher code rates than $(0,k)$- \citep{immink_weakly_1997} and DC free codes \citep{lee_dc-free_2009} for synchronization in noisy channels. Weakly constrained codes allow sequences that violate constraints with low probability and are conceptually similar to unconstrained coding with randomization, considered here as an alternative to constrained coding.

\citet{buzaglo_row-by-row_2017} find that for flash memory, a combination of weakly constrained codes and error-correcting codes can achieve a higher code rate than removing all error-prone sequence patterns, provided the error rate is low. \citet{li_capacity_2019} estimate the capacity of channels in which errors are due to inter-cell inference in specific bit patterns. They model the input sequences as a Markov chain to control the error rate by dictating the probability of writing an error-prone bit pattern. The capacity expression by \citet{li_capacity_2019} characterizes the error regimes in which unconstrained, constrained, and weakly constrained coding are efficient for flash memory. Our approach differs from theirs in that we simply randomize the DNA strands to control the probability of writing error-prone patterns. We do not explore the optimal proportion of error-prone DNA sequences—that is, at a fixed code rate, the optimal proportion of redundancy allocated to constrained and error-correcting coding to achieve the lowest bit error rates—since we are interested in simpler code designs.

In DNA data storage, much research is devoted to finding capacity achieving homopolymer and GC content constrained codes, as existing algorithms either require additional redundancy to implement the constraints \citep{goldman_towards_2013,bornholt_dna-based_2016,grass_robust_2015,blawat_forward_2016,song_codes_2018,wang_construction_2019,nguyen_capacity-approaching_2021}, have high encoding and decoding complexity \citep{immink_properties_2020,schouhamer_immink_design_2018,nguyen_capacity-approaching_2021,liu_capacity-achieving_2022} or suffer from error propagation \citep{erlich_dna_2017,press_hedges_2020}.
Given these extensive research efforts, the objective of this paper is to evaluate the overall efficiency of homopolymer and GC content constrained coding for DNA data storage, with a focus on substitution errors.
In particular, we aim to determine whether a simple code design (unconstrained coding) can achieve code rates comparable to or higher than constrained coding.

\section{Homopolymers}
\label{Homopolymers}

In this section, we compare constrained coding to avoid homopolymers (called runs henceforth) with unconstrained coding. We find that unconstrained coding is more efficient in terms of achievable code rate, unless the increase in substitution rates for nucleotides in runs is very large.

We first describe the underlying channel model, and formally define constrained coding for the run-length as well as unconstrained coding. We then derive the achievable code rates for both coding schemes in the presence of substitution errors. The results are next used to determine at which increases in the substitution error rates constrained coding for the run-length is efficient.

\subsection{Notation and Preliminaries}

We consider the following channel model.

\begin{definition}\textbf{\emph{The run-length varying channel:}}
\label{RLV Channel}
A run-length varying channel maps an input sequence of nucleotides \(\mathbf{X} = X_1 X_2 \cdots X_n\), where each \(X_i\) is from the alphabet \(\{A, C, G, T\}\), to an output sequence of nucleotides \(\mathbf{Y} = Y_1 Y_2 \cdots Y_n\), where each \(Y_i\) is from the same alphabet. The sequence length is denoted by \(n\). The substitution probability \(p_r\) for a nucleotide \(X_i\) is determined by its run-length \(r\), which is the number of consecutive identical nucleotides to which \(X_i\) belongs. The substitution probabilities are symmetric across the nucleotide types, and \(p_r\) is non-decreasing as a function of the run-length, i.e., \(p_r \geq p_{r-1} \geq \cdots \geq p_1 = p\).
\end{definition}

In the run-length varying channel, the probability of substitution is determined by the run-length to which the transmitted nucleotide belongs. In the literature, channels whose error characteristics vary during transmission are also known as channels with random states \citep{gamal_network_2011}.

We define constrained and unconstrained coding for the run-length varying channel as follows.

\begin{definition}
\textbf{\emph{\(m\)-constrained and unconstrained coding:}}
\label{def:constrained_coding}
An \(m\)-constrained code for the run-length varying channel consists of the following:
\begin{itemize}
    \item A set of message indices \(\{1, 2, \ldots, M\}\).
    \item An encoding function \(f_m: \{1, 2, \ldots, M\} \to \mathcal{A}_m\) that maps message indices to codewords in \(\mathcal{A}_{m}\), the subset of all input sequences \(\mathbf{X} \in \{A, C, G, T\}^n\) whose maximum run-length is \(m\).
    \item A decoding function \(g: \{A, C, G, T\}^n \to \{1, 2, \ldots, M\}\) that maps each received output sequence back to a message index.
\end{itemize}
An unconstrained code for the run-length varying channel is an \(\infty\)-constrained code, wherein the encoding function \(f_\infty \triangleq f\) maps message indices to the set of all possible sequences \(\{A, C, G, T\}^n\).
\end{definition}


\subsection{Achievable code rates}
\label{cr}

We derive achievable code rates for \(m\)-constrained and unconstrained coding in an asymptotic regime using a random coding argument.  
The code rate and achievability of a code rate are defined as follows. 
\begin{definition}\textbf{\emph{Code rate:}}
\label{AC}
The code rate for a code $\mathcal{C}$ of size $M$ and length $n$ is defined as:
\begin{equation*}
    R=\frac{\log_2(M)}{n}.
\end{equation*}
For $m$-constrained coding, a code rate $R_c$ is said to be achievable if there exists a sequence of $m$-constrained codes $(\mathcal{C}_m^1, \mathcal{C}_m^2, \cdots)$, where each $\mathcal{C}_m^n \in \mathcal{A}_m$ and the maximum (over all codewords) probability of decoding error approaches zero as the sequence length \(n \to \infty\). Similarly, for unconstrained coding, a code rate $R_u$ is said to be achievable if there exists a sequence of unconstrained codes $(\mathcal{C}^1, \mathcal{C}^2, \ldots)$, where each $\mathcal{C}^n \in \{A,C,G,T\}^n$ and the maximum (over all codewords) probability of decoding error approaches zero as \(n \to \infty\).
\end{definition}

Studying the achievable code rate provides theoretical insight into the trade-off between higher error rates and run-length constraints, without considering practical limitations imposed by finite sequence lengths and specific error-correcting code constructions.

We state achievable code rates for the two coding schemes in Theorem~\ref{th1} and defer the proof to Appendix~\ref{pf1}. 

\begin{theorem}\textbf{\emph{Achievable code rates for \(m\)-constrained coding and unconstrained coding:}}
\label{th1}
Let \(H(p_r)\) be the entropy of a quaternary random variable that retains its state with probability \(1-p_r\) and substitutes to one of the other three states with probability \(p_r/3\):
\[
H(p_r) = - \left((1-p_r) \log_2(1-p_r) + p_r \log_2\left(\frac{p_r}{3}\right)\right).
\]
Define \( q(r) \) as the asymptotic probability that a random nucleotide \( X_i \) in a sequence \( \mathbf{X} \in \{A,C,G,T\}^n \) occurs in a run of length \( r \) and \( q_m(r) \) as the normalized probability for a sequence \( \mathbf{X} \in \mathcal{A}_m \). \\
The following statements hold. 
\begin{enumerate}
    \item For $m$-constrained coding, define $R_c$ as follows: 
\begin{equation}
\label{R_c}
R_c \triangleq H\left(P^{\mathbf{Y}}\right) - \sum_{r=1}^{m} q_m(r) H(p_r), \quad \text{with} \quad q_m(r) = \frac{q(r)}{\sum_{s=1}^m q(s)},
\end{equation}
where \(H\left(P^{\mathbf{Y}}\right) = \lim_{n \to \infty} \frac{1}{n} H\left(\mathbf{Y}\right)\) is the entropy rate of the stochastic process generating output sequences \(\mathbf{Y}\) and $q(r)$ is defined below. Then $R_c$ is an achievable code rate for $m$-constrained coding.
\item For unconstrained coding, define $R_u$ as follows:
\begin{equation}
\label{R_u}
R_u \triangleq 2 - \sum_{r=1}^{n} q(r) H(p_r), \quad \text{with} \quad q(r) = r\left(\frac{1}{4}\right)^{r-1} \left(\frac{3}{4}\right)^2.
\end{equation}
Then $R_u$ is an achievable code rate for unconstrained coding.
\end{enumerate}
\end{theorem}

\subsection{Achievable code rates for different substitution rate increases }
\label{lbr}

Theorem~\ref{th1} specifies the achievable code rates for $m$-constrained and unconstrained coding as a function of the substitution probabilities $p_{r}$. 
Whether $m$-constrained or unconstrained coding is more efficient depends on how the $p_{r}$'s increase as a function of run-length $r$, which is a property of the channel that varies between DNA storage systems. We consider different values of $p_{r} \geq \ldots \geq p_{2} \geq p_{1}=p$ to determine the error regimes in which $m$-constrained coding achieves a higher achievable code rate than unconstrained coding. 
If the substitution probabilities do not increase as a function of run-length, unconstrained coding is clearly preferred. 
In Section~\ref{emp}, we discuss how error rates (substitutions, deletions, and insertions) increase as a function of run-length in existing DNA storage systems.

Let us consider a linear growth model for the substitution rate $p_{r}$:
\begin{align*}
 p_{r} = \min \left(0.75, \alpha \left(r-1\right)+p\right),
  \end{align*}
where $r \in \mathbb{N}$, $\alpha \geq 0$ is a growth factor and $p=p_1$ is a base substitution rate set to $1\%$, consistent with the substitution probabilities of current DNA storage systems that lie between $0.08\%$ and $2.6\%$ \citep{grass_robust_2015, erlich_dna_2017, goldman_towards_2013, antkowiak_low_2020}.
Note that the worst-case substitution rate is $0.75$. If $p_{r}=0.75$, the least amount of information is available (the entropy is maximized), as each nucleotide could be present with equal probability regardless of the observed channel output.

Figure~\ref{fig3} shows the error regimes in which $m$-constrained coding is more efficient than unconstrained coding and vice versa. The error regimes are color-coded based on the associated achievable code rate difference~$R_u-R_c$. The orange-shaded region correspond to the values of $p_{r}$ for which $m$-constrained coding is more efficient. Conversely, the purple-shaded region correspond to the values of $p_{r}$ for which unconstrained coding is more efficient. The gray line indicates similar performance between the two coding schemes.

\begin{figure}[t]
    \centering
\includegraphics[scale=1]{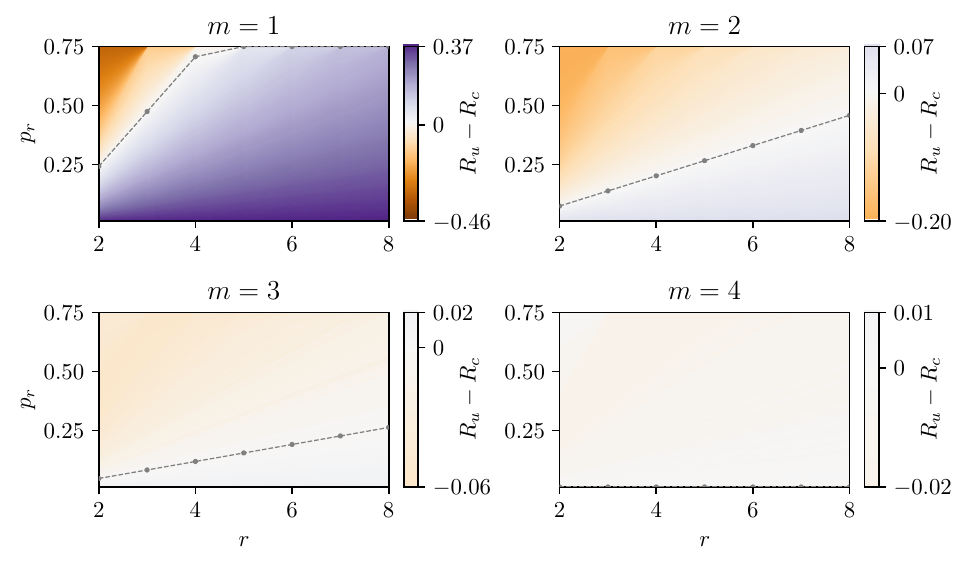}
\captionsetup{labelfont=bf}
    \caption{Error regimes in which $m$-constrained and unconstrained coding achieve a larger code rate. The error regimes are color-coded based on the associated achievable code rate difference~$R_u-R_c$, where the gray line indicates similar performances. }
    \label{fig3}
\end{figure}

Figure~\ref{fig3} indicates that the increase in the substitution rate for nucleotides in runs must be very large for $m$-constrained coding to be efficient. For example, $1$-constrained coding becomes efficient for substitution probabilities \( p_{r>4}=0.75 \), i.e., when nucleotides in runs longer than four are maximally error-prone. Similarly, $3$-constrained coding becomes efficient when nucleotides in a run of length five have a substitution rate fifteen times that of nucleotides in a run of length one.

Weaker constraints require less redundancy, and $m$-constrained coding becomes efficient at smaller increases in the substitution rate. However, the weaker the constraint, the smaller the code rate difference~$R_u-R_c$. Thus, the gain in achievable code rate in the error regimes where $m$-constrained coding is efficient must be weighed against the higher complexity and optimality of the code design. For example, when the maximum run-length is $m=3$ or $m=4$, the gain in achievable code rate is almost negligible, even in the error regimes where $m$-constrained coding is more efficient.

When \( m > 4 \), the entropy of the output process \( H(P^{\mathbf{Y}}) \) for \( m \)-constrained coding approaches two, which is the entropy of the output process for unconstrained coding. Additionally, the probability of reading a sequence with runs longer than four is approximately zero. As a result, the achievable code rates for $m$-constrained and unconstrained coding are approximately equal. In such scenarios, unconstrained coding is preferred due to its simpler code design.

In this section, we discussed linear substitution rate increases. We discuss other growth models and associated error regimes in Appendix~\ref{App C}.

\section{GC-Content}
\label{GC-Content}

In addition to constrained coding for the run-length, many papers propose code designs that balance the GC content of the DNA sequences \citep{chee_linear-time_2019, king_bounds_2003, eidin_constant_2018, cai_correcting_2021}. In this section, we compare constrained coding for the GC content with unconstrained coding in the presence of substitution errors. We find that the differences between the two coding schemes in terms of Gilbert-Varshamov code rate lower bounds are marginal for common substitution error rates and sequence lengths.

Following the approach of the previous section, we first introduce our channel model and formally define constrained coding for the GC content as well as unconstrained coding. We then give Gilbert-Varshamov code rate lower bounds for the two coding schemes. Next, we use the results to determine the substitution error rate increases at which constrained coding for the GC content achieves a larger Gilbert-Varshamov code rate lower bound than unconstrained coding.

\subsection{Notation and preliminaries}


We consider the following channel model.

\begin{definition}\textbf{\emph{The GC content channel:}}
\label{GC Channel}
A GC content channel maps an input sequence of nucleotides \(\mathbf{X} = X_1X_2 \cdots X_n\), where each \(X_i\) is from the alphabet \(\{A, C, G, T\}\), to an output sequence of nucleotides \(\mathbf{Y} = Y_1Y_2 \cdots Y_n\), where each \(Y_i\) is from the same alphabet. The sequence length is denoted by \(n\). The substitution probabilitiy \(p_w\) for any nucleotide \(X_i\) is determined by the sequence's GC content \(w\), defined as the number of nucleotides that are either \(G\) or \(C\) in input sequence \(\mathbf{X}\), with \(0 \leq w \leq n\). The substitution probabilities are symmetric across the nucleotide types, and \(p_w\) is non-decreasing as a function of the imbalance in GC content, satisfying \(p_0 \geq \cdots \geq p_{\lfloor n/2 \rfloor} = p \leq \cdots \leq p_n\).
\end{definition}

In the GC content channel, the substitution probability for each nucleotide is constant during the transmission of a sequence, but varies for different input sequences. In the literature, channels whose error characteristics can vary between transmissions, but are constant for each transmission, are also known as compound channels \citep{gamal_network_2011}.

We define constrained and unconstrained coding for the GC content channel as follows.

\begin{definition}\textbf{\emph{\(\epsilon\)-constrained and unconstrained coding:}}
\label{def:1}
An \(\epsilon\)-constrained code for the GC content channel consists of the following:
\begin{itemize}
  \item A set of message indices \(\{1, 2, \ldots, M\}\).
  \item An encoding function \(f_{\epsilon}: \{1, 2, \ldots, M\} \to \mathcal{S}_{\epsilon}\) that maps message indices to the subset $\mathcal{S}_{\epsilon}$ of all input sequences \( \mathbf{X} \in \{A, C, G, T\}^n\) that have a GC content \( w \) satisfying \( \lceil (0.5 - \epsilon)n \rceil \leq w \leq \lfloor (0.5 + \epsilon)n \rfloor \).
  \item A decoding function \(g: \{A, C, G, T\}^n \to \{1, 2, \ldots, M\}\) that maps each received output sequence back to a message index.
\end{itemize}
An unconstrained code for the GC content channel is an \((0.5, n)\)-constrained code, wherein the encoding function $f_{0.5}\triangleq f$ maps message indices uniformly and independently to the set of all possible sequences \( \{A, C, G, T\}^n\).
\end{definition}

In an asymptotic regime where \(n \to \infty\), the GC content of random sequences stabilizes at 50\%. Thus, \(\epsilon\)-constrained and unconstrained coding become approximately equal, and a comparison between them is meaningless. Instead, we focus on the finite-length regime and compare the coding schemes using Gilbert-Varshamov code rate lower bounds. Alternatively, one could study \(\epsilon\)-constrained and unconstrained coding in an asymptotic regime and consider local rather than global GC content constraints.

\subsection{Gilbert-Varshamov code rate lower bounds}
\label{crGC}
We analyze the Gilbert-Varshamov bound to derive code rate lower bounds for \(\epsilon\)-constrained and unconstrained coding. Recall Definition \ref{AC} of the code rate: 
\[
R = \frac{1}{n} \log_2 M_q(n, d),
\]
where the size \( M \) of the code is now a function of sequence length \( n \), minimum Hamming distance \( d \), and alphabet size \( q \). The minimum Hamming distance \( d \) is the smallest number of positions at which any two codewords in the code can differ, and thus determines the number of errors the code can correct.

The Gilbert-Varshamov bound provides a theoretical lower bound on \(M_q(n, d)\). For \(\epsilon\)-constrained coding, we extend this bound to additionally consider the constraint \(\epsilon\), thus providing a theoretical lower bound on the maximum number \(M_q(\epsilon,n,d)\) of distinct codewords an \(\epsilon\)-constrained code can contain.

For unconstrained coding, the set of all possible sequences is the entire space $\{A,C,G,T\}^n$, and the maximum number of distinct codewords is not constrained by \(\epsilon=0.5\), such that \(M_q(0.5, n,d)=M_q(n,d)\). The Gilbert-Varshamov bound for unconstrained coding is given by the ratio of the total number of possible sequences to the volume of a Hamming ball of radius \(d-1\), as stated in the following Theorem \ref{th:2}.

\begin{theorem} \textbf{\emph{Gilbert-Varshamov bound for unconstrained coding \citep{marcus_introduction_2001}:}}
\label{th:2}
The maximum size, \(M_q(n,d)\), of an unconstrained code of length \(n\), minimum Hamming distance \(d\) (where \(0 \leq d \leq n\)), and alphabet size \(q=4\), satisfies:
\begin{equation}
    M_q(n,d) \geq \frac{4^n}{\sum_{i = 0}^{d - 1} \binom{n}{i} 3^i},
    \label{GV_u}
\end{equation}
where the denominator is the Hamming ball volume, defined as the number of sequences \(\mathbf{X}' \in \{A, C, G, T\}^n\) within a Hamming distance \(d(\mathbf{X}, \mathbf{X}') \leq d-1\) from any center sequence \(\mathbf{X} \in \{A, C, G, T\}^n\).

This results in the following code rate lower bound for unconstrained coding:
\begin{equation*}
    R_u \geq R^l_u \triangleq 2 - \frac{1}{n} \log_2 \left( a \right),
\end{equation*}
where \(a=\sum_{i = 0}^{d - 1} \binom{n}{i} 3^i\).
\end{theorem}

In unconstrained coding, the Hamming ball volume is independent of the center sequence (the denominator of Equation \eqref{GV_u} does not depend on sequence \(\mathbf{X}\)). In contrast, for most constrained codes, the Hamming ball volume varies with the center sequence. \citet{gu_generalized_1993} address this and show that the Gilbert-Varshamov bound for constrained codes is calculated as the ratio of the number of sequences that satisfy the constraint to the average Hamming ball volume across all sequences in the constrained space.

\citet{king_bounds_2004} derives an expression for the Gilbert-Varshamov bound for codes with constant GC content. To analyze the Gilbert-Varshamov bound for \(\epsilon\)-constrained coding, we extend the result of \citet{king_bounds_2004} to codes whose GC content can vary within a predefined range determined by the \(\epsilon\) constraint. We state the expression for the Hamming ball volume in constrained spaces \(\mathcal{S}_{\epsilon}\) in Lemma~\ref{HB} and defer the proof to Appendix~\ref{pf2}.

\begin{lemma} \textbf{\emph{Hamming ball volume in constrained space \(\mathcal{S}_{\epsilon}\):}}
\label{HB}
Define \(V_{\epsilon}(\mathbf{X}) = \{\mathbf{X}' \in \mathcal{S}_{\epsilon} : d(\mathbf{X},\mathbf{X}') \leq d - 1\}\) as the Hamming ball centered at sequence \(\mathbf{X} \in \mathcal{S}_{\epsilon}\). The volume of \(V_{\epsilon}(\mathbf{X})\) is:
\begin{equation}
\label{eq:4}
|V_{\epsilon}(\mathbf{X})| = \sum_{r=0}^{d-1} \sum_{\Delta=\max(\lceil(0.5 - \epsilon)n\rceil\! - w, -r)}^{\min(\lfloor(0.5 + \epsilon)n\rfloor\! - w, r)} \sum_{i_+=\max(0, \Delta)}^{\min(\Delta + w, r)} \binom{w}{i_+\! -\! \Delta} \binom{n-w}{i_+} \binom{n-2i_+\! +\! \Delta}{r - 2i_+\! +\! \Delta} 2^{2i_+\! -\! \Delta}.
\end{equation}
\end{lemma}

Next, we use the expression for the Hamming ball volume in constrained spaces \(\mathcal{S}_{\epsilon}\) to derive a Gilbert-Varshamov bound for \(\epsilon\)-constrained coding in Theorem~\ref{th:3}. Note that Equation~\eqref{eq:4} only depends on GC content \(w\), not on the specific center sequence \(\mathbf{X}\). We abuse notation and write \(V_{\epsilon}(w)\) for the Hamming ball centered at a sequence \(\mathbf{X} \in \mathcal{S}_{\epsilon}\) with GC content \(w\). Moreover, we write \(w \in \mathcal{S}_{\epsilon}\) to denote GC contents \(w\) that satisfy the set constraint \( \lceil (0.5 - \epsilon)n \rceil \leq w \leq \lfloor (0.5 + \epsilon)n \rfloor\).

\begin{theorem}\textbf{\emph{Gilbert-Varshamov bound for $\epsilon$-constrained coding:}}
\label{th:3}
The maximum size, \(M_q(\epsilon, n,d)\), of an \(\epsilon\)-constrained code of length \(n\), minimum Hamming distance \(d\) (where \(0 \leq d \leq n\)), and alphabet size \(q = 4\), satisfies:
\begin{equation}
\label{GV}   
M_q(\epsilon, n,d) \geq \frac{\sum_{w \in \mathcal{S}_{\epsilon}} \binom{n}{w} 2^n}{\sum_{w \in \mathcal{S}_{\epsilon}} q_{\epsilon}(w) |V_{\epsilon}(w)|},
\end{equation}
where the probability \(q_{\epsilon}(w)\) is the proportion of sequences in \(\mathcal{S}_{\epsilon}\) with GC content \(w\), for which an expression is given in Lemma~\ref{GC distr}. An expression for the Hamming ball volume \(|V_{\epsilon}(w)|\) is given in Lemma~\ref{HB}.
This results in the following code rate lower bound for \(\epsilon\)-constrained coding:
\begin{equation*}
  R_c \geq R^l_c \triangleq 1 - \frac{\log_2(a)}{n} + \frac{\log_2(b)}{n}, 
\end{equation*}
where \(a = \sum_{w \in \mathcal{S}_{\epsilon} }q_{\epsilon}(w) |V_{\epsilon}(w)|\) and  \(b=\sum_{w \in \mathcal{S}_{\epsilon}} \binom{n}{w} 2^n\).
\end{theorem}

\begin{proof}
Equation~\eqref{GV} follows from the result by \citet{gu_generalized_1993}, who show that the Gilbert-Varshamov bound for constrained codes is the ratio of the number of sequences that satisfy the constraint to the average Hamming ball volume. The numerator is the size of constrained space \( \mathcal{S}_{\epsilon} \), obtained by summing over all sequences where the GC content \( w \) is within the range determined by the $\epsilon$ constraint. The coefficient \( \binom{n}{w} \) calculates the number of ways to select \( w \) positions from \( n \) available positions for the nucleotides \(\{G, C\}\). The factor \( 2^n \) accounts for all binary choices at both the selected \(\{G, C\}\) positions and the remaining \(\{A, T\}\) positions. The denominator is the average Hamming ball volume \( |\overline{V}| = \sum_{\mathbf{X} \in \mathcal{S}_{\epsilon}} |V_{\epsilon}(\mathbf{X})|/|\mathcal{S}_{\epsilon}| \). An expression for $|V_{\epsilon}(\mathbf{X})|$ is given in Lemma \ref{HB}. The Hamming ball volume \( |V_{\epsilon}(\mathbf{X})| \) depends only on the GC content \( w \) of the center sequence, allowing us to rewrite the average Hamming ball volume as a weighted sum over all permissible GC contents \( w \), where the weights \( q_{\epsilon}(w) \) are the proportions of sequences with GC content \( w \) in constrained space \( \mathcal{S}_{\epsilon} \), for which an expression is given in Lemma~\ref{GC distr}.
\end{proof}

\subsection{Code rate lower bounds for different substitution rate increases}

\label{lbc}
Theorem~\ref{th:3} states the Gilbert-Varshamov code rate lower bounds for \(\epsilon\)-constrained and unconstrained coding as a function of the minimum Hamming distance $d$. To compare \(\epsilon\)-constrained with unconstrained coding, we calculate the code rate lower bounds for minimum Hamming distances $d$ corresponding to the expected number of errors for each coding scheme.

The expected number of errors depends on how the substitution probabilities $p_w$ increase with imbalances in GC content, which is a property of the DNA storage system. Therefore, we consider different values of $p_{0} \geq \cdots \geq p_{\lfloor n/2 \rfloor}=p \leq \cdots \leq p_{n}$. In Section~\ref{emp}, we discuss how substitution, deletion, and insertion rates correlate with imbalances in GC content in existing DNA storage systems.

We consider a parabolic growth model for the substitution rates $p_{w}$:
\begin{align*}
    p_{w}= \min \left(0.75, \alpha\left(\frac{w}{n}-0.5\right)^2+p \right),
\end{align*}
where $0\leq w\leq n$, $\alpha>0$ is a growth factor and $p=p_{\lfloor n/2 \rfloor}$ is a base error probability set to $1\%$. We select a parabolic growth model because it represents the worst-case scenario for unconstrained coding, wherein both low and high GC contents are associated with higher substitution error rates. The maximum substitution rate is $0.75$ because, at this rate, the output is statistically independent of the input.

The expected number of substitution errors for each coding scheme is given by $\bar{p}n$, where $\bar{p}$ is the average substitution probability. For \(\epsilon\)-constrained coding, $\bar{p}$ is computed as the weighted sum $\bar{p} = \sum q_{\epsilon}(w)p_w$, where the weights $q_{\epsilon}(w)$ are the proportion of sequences with GC content $w$ in constrained space \(\mathcal{S}_{\epsilon}\). Similarly, for unconstrained coding, $\bar{p}$ is computed as $\bar{p} = \sum q(w)p_w$, where $q_{0.5}(w)\triangleq q(w)$ is the proportion of sequences with GC content $w$ from all possible sequences $\{A,C,G,T\}^n$. The probability distributions $q_{\epsilon}(w)$ and $q(w)$ are given in the following Lemma \ref{GC distr}.

\begin{lemma} \textbf{\emph{Proportion of sequences with GC content $w$:}}
\label{GC distr}
The proportion \(q_{\epsilon}(w)\) of sequences with GC content \(w\) within constrained space \(\mathcal{S}_{\epsilon}\) is:
\begin{equation*}
    q_{\epsilon}(w) = \frac{q(w)}{\sum_{s \in \mathcal{S}_{\epsilon}} q(s)},
\end{equation*}
where \(q(w)\) represents the proportion of sequences with GC content \(w\) in the total sequence space \(\{A,C,G,T\}^n\), given by:
\begin{equation*}
   q(w) = \frac{1}{2^n} \binom{n}{w}.
\end{equation*}
\end{lemma}

\begin{proof}
The proportion \( q(w) \) is calculated by counting the number of sequences with GC content \( w \), which is given by \( 2^w \binom{n}{w} \), and then dividing this count by the total number of sequences, \( 4^n \). The proportion \( q_{\epsilon}(w) \) adjusts \( q(w) \) for the constrained subset \(\mathcal{S}_{\epsilon}\) by dividing \( q(w) \) by the sum of \( q(w) \) over all \( w \in \mathcal{S}_{\epsilon} \).
\end{proof}

\begin{figure}[t]
    \centering
\includegraphics[scale=1]{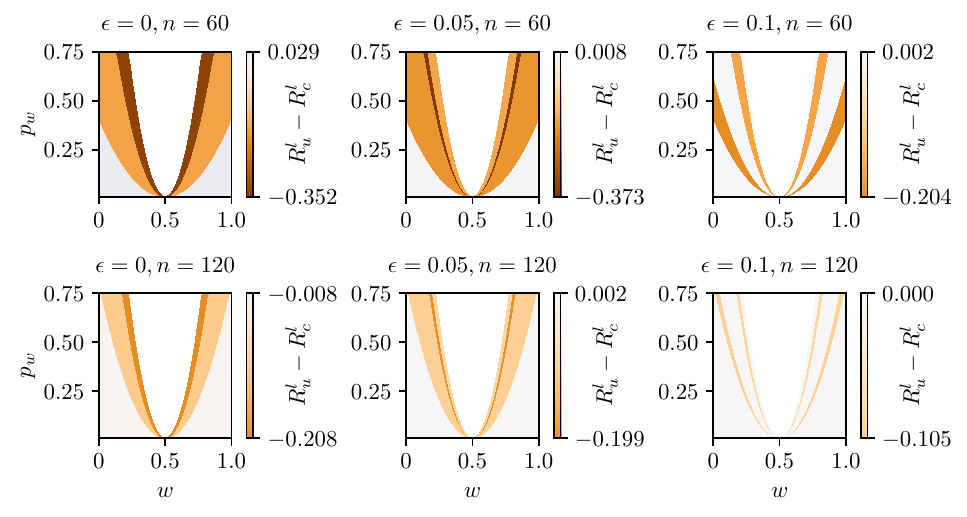}
\captionsetup{labelfont=bf}
    \caption{
    Error regimes in which \(\epsilon\)-constrained and unconstrained coding achieve a larger Gilbert-Varshamov code rate lower bound. The error regimes are color-coded based on the code rate lower bound difference~$R_u^l-R_c^l$, where gray indicates similar performances.  
    }
    \label{fig4}
\end{figure}

Figure~\ref{fig4} characterizes the error regimes in which \(\epsilon\)-constrained coding achieves a larger Gilbert-Varshamov code rate lower bound than unconstrained coding and vice versa for growth factors \(0 \leq \alpha < 10\). For growth factors \(\alpha \geq 10\), the substitution rate increases are far from what is expected in practice. For example, for \(\alpha=10\) and \(n=120\), nucleotides in sequences with GC content less than \(30\%\) or more than \(70\%\) are maximally error-prone. The error regimes are color-coded based on the associated Gilbert-Varshamov code rate lower bound difference \(R_u^l-R_c^l\). Orange indicates the region where \(\epsilon\)-constrained coding achieves a larger code rate lower bound, purple where unconstrained coding achieves a larger code rate lower bound. Gray indicates similar performances.

Overall, for common substitution rate increases, the performance differences between the two coding schemes are marginal. For example, \(0\)-constrained coding achieves a higher code rate lower bound than unconstrained coding (gain \(\approx 0.18\)), for growth factors \(1.6 \leq \alpha \leq 5.6\) and $n=60$. For \(1.6 \leq \alpha \leq 5.6\), the substitution rate for all nucleotides in sequences with \(30\%\) or \(70\%\) GC content is at least five times higher than that in balanced sequences.

The maximum gains in Gilbert-Varshamov code rate lower bounds for \(0\)- and \(0.05\)-constrained coding are approximately \(0.35\) and \(0.37\). These gains are achieved for growth factors exceeding \(5.6\). Under such conditions, all nucleotides in sequences with less than \(15\%\) or more than \(85\%\) GC content become maximally error-prone, an unlikely scenario in practical applications. Therefore, given the marginal differences at common substitution rates, unconstrained coding is preferred over \(\epsilon\)-constrained coding for its simpler code design.

\section{Empirical error analysis}
\label{emp}

Our theoretical results indicate that the increase in substitution rates for nucleotides in runs and in sequences with unbalanced GC content must be very large for constrained coding to be efficient.
To understand in which error regimes current DNA storage systems operate, we empirically analyze how the substitution rates increase as a function of run-length and GC content.
We find that existing DNA storage systems \citep{srinivasavaradhan_trellis_2021,antkowiak_low_2020, meiser_reading_2020, gimpel_digital_2023} lie in error regimes in which unconstrained coding is more efficient than constrained coding for substitution errors.

While we lack theoretical results for the efficiency of constrained coding for insertions and deletions, we also explore how insertion and deletion rates change as a function of run-length and GC content imbalance.

We limit our analysis to runs up to length six and GC content between approximately $35\%$ to $65\%$. In this range, we observe no significant increase in substitution (\(p^S\)), insertion (\(p^I\)), and deletion (\(p^D\)) rates. In random codebooks, runs longer than six and sequences with GC content above $65\%$ or below $35\%$ occur infrequently. Consequently, their impact on the average error rates is minimal. For example, in the dataset by \citet{srinivasavaradhan_trellis_2021}, among the 10,000 randomized sequences (each 110 nucleotides long), runs of seven and eight nucleotides occur only 156 and 16 times, respectively, with no runs longer than eight. The distribution of run-length and GC content for all experiments is shown in Figures~\ref{App12} and~\ref{App13} in Appendix~\ref{App D}, respectively.

Several factors influence the error rates in practical DNA storage systems. The main distinction is usually made between the synthesis and sequencing technologies used. Table~\ref{table1} provides an overview of the technologies, constraints, and sequence designs in the DNA storage systems by \citet{srinivasavaradhan_trellis_2021}; \citet{antkowiak_low_2020}; Netflix (data available at \href{http://d3u0hl24lyh7px.cloudfront.net/encoded_withprimers.txt}{Netflix dataset}); and \citet{gimpel_digital_2023}.

\begin{table}[t]
\centering
\resizebox{\textwidth}{!}{%
\begin{tabular}{|c|c|c|c|c|c|}
\hline
\textbf{Dataset} & \textbf{Length} & \textbf{Number}  & \textbf{Synthesis} & \textbf{ Sequencing} & \textbf{Constraint} \\ 
\hline
\hline
Srinivasavaradhan et al. & 110 & 10,000
& Twist & Nanopore & - \\ 
Antkowiak et al. & 60 & 16,383
& Photolithographic & Illumina & - \\
Netflix & 105+20  & 3,900,000
& Twist & Illumina & - \\ 
Gimpel et al. \uppercase\expandafter{\romannumeral1} & $102+41$ & $12,472$
& Genscript & Illumina & - \\ 
Gimpel et al. \uppercase\expandafter{\romannumeral2} & $117+41$ & $12,402$
& Genscript & Illumina &  $\epsilon=0$ \\
Gimpel et al. \uppercase\expandafter{\romannumeral3} & $108+41$& $12,000$
& Twist & Illumina & - \\
Gimpel et al. \uppercase\expandafter{\romannumeral4} &$108+41$ & $12,000$
& Twist & Illumina &   $\epsilon=0$\\ \hline
 
\end{tabular}%
}
\captionsetup{labelfont=bf}
\caption{\textbf{Dataset characteristics.} The sequence lengths range from $60-117$ nucleotides with information, plus additional nucleotides as primers when sequencing is done with Illumina. \citet{antkowiak_low_2020} (code available at \href{https://github.com/MLI-lab/noisy_dna_data_storage}{noisy\_dna\_data\_storage}), the Netflix-Pool dataset (code available at \href{https://github.com/reinhardh/dna_rs_coding}{dna\_rs\_coding}) and  \citet{gimpel_digital_2023} Experiments~\uppercase\expandafter{\romannumeral1} follow the encoding described by \citet{meiser_reading_2020}. 
In \citet{gimpel_digital_2023} Experiments~\uppercase\expandafter{\romannumeral2}-\uppercase\expandafter{\romannumeral4}, random sequences without indices are used.}
\label{table1}
\end{table}

Figure~\ref{fig7} summarizes the error rates as a function of run-length. We find no increase in the insertion rate for all DNA storage systems considered. In the dataset by \citet{srinivasavaradhan_trellis_2021}, we estimate an exponential increase in the deletion rate for nucleotides in runs, likely due to Nanopore sequencing. Unlike Illumina sequencing, which reads nucleotides one at a time, Nanopore sequencing identifies bases via conductivity changes as nucleotide blocks pass through the nanopore. This process can make it difficult to accurately determine run-lengths from amplified signals, which can lead to an increase in the deletion rate for runs. In the dataset by \citet{antkowiak_low_2020}, which uses lower-cost, higher-error-rate photolithographic synthesis, we observe a minor logarithmic increase in the substitution rate with run-length, and interestingly, a linear decrease in the deletion rate for nucleotides in runs. However, the increase in the substitution rate is far from the error regimes established in Section~\ref{lbr}, where constrained coding for the run-length is efficient. In all remaining datasets, we find no correlation between run-length and the substitution or deletion rates.

\begin{figure}[t]
\centering
\begin{tikzpicture}
\begin{groupplot}[
    group style={
        group size=4 by 6,  
        vertical sep=2.2cm,   
        horizontal sep=1.2cm  
    },
    error bars/y dir=both,
    error bars/y explicit,
    width=4.3cm,
    height=2.6cm,
    scaled y ticks = false, 
    tick pos=left, 
    tick pos=bottom,
    xmax=6, 
    xmin=1
]

\nextgroupplot[ title={Srinivasavaradhan et al.}, ylabel={\% $p_{r}^S$}, xtick={2, 4,6}, xticklabels={$ $, $ $, $ $},  ytick={0, 0.035}, yticklabels={$0$, $3.5$}]
\addplot [
    mark=o,
    mark size=0.1,
    opacity=0.8,
    line width=0.6pt,
    color=Orchid,
    dashed,
    forget plot
] table [
    x=length,
    y=weighted_rate_avg
] {Srinivasavaradhan_sub.dat};
\addplot [
    only marks,
    opacity=0.8,
    color=Orchid,
    mark size=1.5,
    error bars/.cd,
    y dir=both,
    y explicit
] table [
    x=length,
    y=weighted_rate_avg,
    y error=weighted_rate_std
] {Srinivasavaradhan_sub.dat};

\nextgroupplot[ title={Antkowiak et al.}, ylabel style={align=center, text depth=0ex, yshift=0.5cm}, xtick={2, 4,6}, xticklabels={$ $, $ $, $ $}, ytick={0, 0.05,0.1}, yticklabels={$0$,$5$, $10$}]
\addplot [
    mark=o,
    mark size=0.1,
    opacity=0.8,
    line width=0.6pt,
    color=Peach,
    dashed,
    forget plot
] table [
    x=length,
    y=weighted_rate_avg
] {Antkowiak_sub.dat};
\addplot [
    only marks,
    opacity=0.8,
    color=Peach,
    mark size=1.5,
    error bars/.cd,
    y dir=both,
    y explicit
] table [
    x=length,
    y=weighted_rate_avg,
    y error=weighted_rate_std
] {Antkowiak_sub.dat};

\nextgroupplot[ title={Netflix-Pool}, ylabel style={align=center, text depth=0ex, yshift=0.5cm}, xtick={2, 4,6}, xticklabels={$ $, $ $, $ $}, ytick={0, 0.04}, yticklabels={$0$, $4$}]
\addplot [
    mark=o,
    mark size=0.1,
    opacity=0.8,
    line width=0.6pt,
    color=Orchid,
    dashed,
    forget plot
] table [
    x=length,
    y=weighted_rate_avg
] {netflix_sub.dat};
\addplot [
    only marks,
    opacity=0.8,
    color=Orchid,
    mark size=1.5,
    error bars/.cd,
    y dir=both,
    y explicit
] table [
    x=length,
    y=weighted_rate_avg,
    y error=weighted_rate_std
] {netflix_sub.dat};

\nextgroupplot[title={Gimpel et al. \uppercase\expandafter{\romannumeral1}}, xtick={2, 4,6}, xticklabels={$ $, $ $, $ $},ytick={0, 0.015}, yticklabels={$0$, $1.5$} ]
\addplot [
    mark=o,
    mark size=0.1,
    opacity=0.8,
    line width=0.6pt,
    color=Orchid,
    dashed,
    forget plot
] table [
    x=length,
    y=weighted_rate_avg
] {GGall_sub.dat};
\addplot [
    only marks,
    opacity=0.8,
    color=Orchid,
    mark size=1.5,
    error bars/.cd,
    y dir=both,
    y explicit
]table [
    x=length,
    y=weighted_rate_avg,
    y error=weighted_rate_std
] {GGall_sub.dat};

\nextgroupplot[ yshift=2.1cm, xtick={2, 4,6}, ylabel={\% $p_{r}^I$}, xticklabels={$ $, $ $, $ $}, ytick={0, 0.035}, yticklabels={$0$, $3.5$}]
\addplot [
    mark=o,
    mark size=0.1,
    opacity=0.8,
    line width=0.6pt,
    color=Orchid,
    dashed,
    forget plot
] table [
    x=length,
    y=weighted_rate_avg
] {Srinivasavaradhan_ins.dat};
\addplot [
    only marks,
    opacity=0.8,
    color=Orchid,
    mark size=1.5,
    error bars/.cd,
    y dir=both,
    y explicit
]table [
    x=length,
    y=weighted_rate_avg,
    y error=weighted_rate_std
] {Srinivasavaradhan_ins.dat};

\nextgroupplot[ yshift=2.1cm, xtick={2, 4,6}, xticklabels={$ $, $ $, $ $}, ytick={0.02, 0.04}, yticklabels={$2$, $4$}]
\addplot [
    mark=o,
    mark size=0.1,
    opacity=0.8,
    line width=0.6pt,
    color=Orchid,
    dashed,
    forget plot
] table [
    x=length,
    y=weighted_rate_avg
] {Antkowiak_ins.dat};
\addplot [
    only marks,
    opacity=0.8,
    color=Orchid,
    mark size=1.5,
    error bars/.cd,
    y dir=both,
    y explicit
]table [
    x=length,
    y=weighted_rate_avg,
    y error=weighted_rate_std
] {Antkowiak_ins.dat};

\nextgroupplot[ yshift=2.1cm, ylabel style={align=center, text depth=0ex, yshift=0.5cm}, xtick={2, 4,6}, xticklabels={$ $, $ $, $ $}, ytick={0,0.006}, yticklabels={$0$, $0.6$}]
\addplot [
    mark=o,
    mark size=0.1,
    opacity=0.8,
    line width=0.6pt,
    color=Orchid,
    dashed,
    forget plot
]table [
    x=length,
    y=weighted_rate_avg
] {netflix_ins.dat};
\addplot [
    only marks,
    opacity=0.8,
    color=Orchid,
    mark size=1.5,
    error bars/.cd,
    y dir=both,
    y explicit
] table [
    x=length,
    y=weighted_rate_avg,
    y error=weighted_rate_std,
] {netflix_ins.dat};

\nextgroupplot[yshift=2.1cm, xtick={2, 4,6}, xticklabels={$ $, $ $, $ $}, ytick={0, 0.003}, yticklabels={$0$, $0.3$}]
\addplot [
    mark=o,
    mark size=0.1,
    opacity=0.8,
    line width=0.6pt,
    color=Orchid,
    dashed,
    forget plot
] table [
    x=length,
    y=weighted_rate_avg
] {GGall_ins.dat};
\addplot [
    only marks,
    opacity=0.8,
    color=Orchid,
    mark size=1.5,
    error bars/.cd,
    y dir=both,
    y explicit
] table [
    x=length,
    y=weighted_rate_avg,
    y error=weighted_rate_std
] {GGall_ins.dat};


\nextgroupplot[yshift=2.1cm, xlabel={run-length $r$}, ylabel={\% $p_{r}^D$}, ytick={0, 0.075,0.15}, yticklabels={$0$, $7.5$, $15$}]
\addplot [
    mark=o,
    mark size=0.1,
    opacity=1,
    line width=0.6pt,
    color=Peach,
    dashed,
    forget plot
] table [
    x=length,
    y=weighted_rate_avg
] {Srinivasavaradhan_del.dat};
\addplot [
    only marks,
    opacity=0.8,
    color=Peach,
    mark size=1.5,
    error bars/.cd,
    y dir=both,
    y explicit
] table [
    x=length,
    y=weighted_rate_avg,
    y error=weighted_rate_std
] {Srinivasavaradhan_del.dat};

\nextgroupplot[yshift=2.1cm, xlabel={run-length $r$},ytick={0.05,0.1,  0.15}, yticklabels={$5$, $10$, $15$}]
\addplot [
    mark=o,
    mark size=0.1,
    opacity=1,
    line width=0.6pt,
    color=Peach,
    dashed,
    forget plot
] table [
    x=length,
    y=weighted_rate_avg
] {Antkowiak_del.dat};
\addplot [
    only marks,
    opacity=0.8,
    color=Peach,
    mark size=1.5,
    error bars/.cd,
    y dir=both,
    y explicit
] table [
    x=length,
    y=weighted_rate_avg,
    y error=weighted_rate_std
] {Antkowiak_del.dat};

\nextgroupplot[yshift=2.1cm, ylabel style={align=center, text depth=0ex, yshift=0.5cm}, xlabel={run-length $r$}, ytick={0,0.02}, yticklabels={$0$, $2$} ]
\addplot [
    mark=o,
    mark size=0.1,
    opacity=0.8,
    line width=0.6pt,
    color=Orchid,
    dashed,
    forget plot
] table [
    x=length,
    y=weighted_rate_avg
] {netflix_del.dat};
\addplot [
    only marks,
    opacity=0.8,
    color=Orchid,
    mark size=1.5,
    error bars/.cd,
    y dir=both,
    y explicit
] table [
    x=length,
    y=weighted_rate_avg,
    y error=weighted_rate_std
] {netflix_del.dat};

\nextgroupplot[yshift=2.1cm, xlabel={run-length $r$},  ytick={0, 0.02}, yticklabels={$0$, $2$}]
\addplot [
    mark=o,
    mark size=0.1,
    opacity=0.8,
    line width=0.6pt,
    color=Orchid,
    dashed,
    forget plot
] table [
    x=length,
    y=weighted_rate_avg
] {GGall_del.dat};
\addplot [
    only marks,
    opacity=0.8,
    color=Orchid,
    mark size=1.5,
    error bars/.cd,
    y dir=both,
    y explicit
] table [
    x=length,
    y=weighted_rate_avg,
    y error=weighted_rate_std
] {GGall_del.dat};

\nextgroupplot[title={Gimpel et al. \uppercase\expandafter{\romannumeral2}},ylabel={\% $p_{r}^S$}, ylabel style={align=center, text depth=0ex, yshift=0.5cm}, xtick={2, 4,6}, xticklabels={$ $, $ $, $ $},  ytick={0, 0.01}, yticklabels={$0$, $1$}]
\addplot [
    mark=o,
    mark size=0.1,
    opacity=0.8,
    line width=0.6pt,
    color=Orchid,
    dashed,
    forget plot
] table [
    x=length,
    y=weighted_rate_avg
] {GGfix_sub.dat};
\addplot [
    only marks,
    opacity=0.8,
    color=Orchid,
    mark size=1.5,
    error bars/.cd,
    y dir=both,
    y explicit
] table [
    x=length,
    y=weighted_rate_avg,
    y error=weighted_rate_std
] {GGfix_sub.dat};

\nextgroupplot[title={Gimpel et al. \uppercase\expandafter{\romannumeral3}},  xtick={2, 4,6}, xticklabels={$ $, $ $, $ $}, ytick={0, 0.01}, yticklabels={$0$, $1$}]
\addplot [
    mark=o,
    mark size=0.1,
    opacity=0.8,
    line width=0.6pt,
    color=Orchid,
    dashed,
    forget plot
] table [
    x=length,
    y=weighted_rate_avg
] {GTall_sub.dat};
\addplot [
    only marks,
    opacity=0.8,
    color=Orchid,
    mark size=1.5,
    error bars/.cd,
    y dir=both,
    y explicit
] table [
    x=length,
    y=weighted_rate_avg,
    y error=weighted_rate_std
] {GTall_sub.dat};

\nextgroupplot[title={Gimpel et al. \uppercase\expandafter{\romannumeral4}},  xtick={2, 4,6}, xticklabels={$ $, $ $, $ $}, ytick={0, 0.01}, yticklabels={$0$, $ 1$}]
\addplot [
    mark=o,
    mark size=0.1,
    opacity=0.8,
    line width=0.6pt,
    color=Orchid,
    dashed,
    forget plot
] table [
    x=length,
    y=weighted_rate_avg
] {GTfix_sub.dat};
\addplot [
    only marks,
    opacity=0.8,
    color=Orchid,
    mark size=1.5,
    error bars/.cd,
    y dir=both,
    y explicit
] table [
    x=length,
    y=weighted_rate_avg,
    y error=weighted_rate_std
] {GTfix_sub.dat};

\nextgroupplot[hide axis]

\nextgroupplot[yshift=2.1cm, ylabel={\% $p_{r}^I$}, ylabel style={align=center, text depth=0ex, yshift=0.2cm}, xtick={2, 4,6}, xticklabels={$ $, $ $, $ $}, ytick={0,0.003}, yticklabels={$0$, $0.3$}]
\addplot [
    mark=o,
    mark size=0.1,
    opacity=0.8,
    line width=0.6pt,
    color=Orchid,
    dashed,
    forget plot
] table [
    x=length,
    y=weighted_rate_avg
] {GGfix_ins.dat};
\addplot [
    only marks,
    opacity=0.8,
    color=Orchid,
    mark size=1.5,
    error bars/.cd,
    y dir=both,
    y explicit
] table [
    x=length,
    y=weighted_rate_avg,
    y error=weighted_rate_std
] {GGfix_ins.dat};

\nextgroupplot[ yshift=2.1cm, xtick={2, 4,6}, xticklabels={$ $, $ $, $ $}, ytick={0,0.001}, yticklabels={$0$, $0.1$}]
\addplot [
    mark=o,
    mark size=0.1,
    opacity=0.8,
    line width=0.6pt,
    color=Orchid,
    dashed,
    forget plot
] table [
    x=length,
    y=weighted_rate_avg
] {GTall_ins.dat};
\addplot [
    only marks,
    opacity=0.8,
    color=Orchid,
    mark size=1.5,
    error bars/.cd,
    y dir=both,
    y explicit
] table [
    x=length,
    y=weighted_rate_avg,
    y error=weighted_rate_std
] {GTall_ins.dat};

\nextgroupplot[ yshift=2.1cm,  xtick={2, 4,6}, xticklabels={$ $, $ $, $ $}, ytick={0,0.002}, yticklabels={$0$, $0.2$}]
\addplot [
    mark=o,
    mark size=0.1,
    opacity=0.8,
    line width=0.6pt,
    color=Orchid,
    dashed,
    forget plot
] table [
    x=length,
    y=weighted_rate_avg
] {GTfix_ins.dat};
\addplot [
    only marks,
    opacity=0.8,
    color=Orchid,
    mark size=1.5,
    error bars/.cd,
    y dir=both,
    y explicit
] table [
    x=length,
    y=weighted_rate_avg,
    y error=weighted_rate_std
] {GTfix_ins.dat};

\nextgroupplot[hide axis]

\nextgroupplot[yshift=2.1cm, xlabel={run-length $r$},ylabel={\% $p_{r}^D$}, ylabel style={align=center, text depth=0ex, yshift=0.5cm},  ytick={0,0.03}, yticklabels={$0$, $3$}]
\addplot [
    mark=o,
    mark size=0.1,
    opacity=0.8,
    line width=0.6pt,
    color=Orchid,
    dashed,
    forget plot
] table [
    x=length,
    y=weighted_rate_avg
] {GGfix_del.dat};
\addplot [
    only marks,
    opacity=0.8,
    color=Orchid,
    mark size=1.5,
    error bars/.cd,
    y dir=both,
    y explicit
] table [
    x=length,
    y=weighted_rate_avg,
    y error=weighted_rate_std
] {GGfix_del.dat};

\nextgroupplot[yshift=2.1cm, xlabel={run-length $r$}, ytick={0,0.004}, yticklabels={$0$, $0.4$}]
\addplot [
    mark=o,
    mark size=0.1,
    opacity=0.8,
    line width=0.6pt,
    color=Orchid,
    dashed,
    forget plot
] table [
    x=length,
    y=weighted_rate_avg
] {GTall_del.dat};
\addplot [
    only marks,
    opacity=0.8,
    color=Orchid,
    mark size=1.5,
    error bars/.cd,
    y dir=both,
    y explicit
] table [
    x=length,
    y=weighted_rate_avg,
    y error=weighted_rate_std
] {GTall_del.dat};

\nextgroupplot[ yshift=2.1cm, xlabel={run-length $r$},ytick={0,0.03}, yticklabels={$0$, $3$}]
\addplot [
    mark=o,
    mark size=0.1,
    opacity=0.8,
    line width=0.6pt,
    color=Orchid,
    dashed,
    forget plot
] table [
    x=length,
    y=weighted_rate_avg
] {GTfix_del.dat};
\addplot [
    only marks,
    opacity=0.8,
    color=Orchid,
    mark size=1.5,
    error bars/.cd,
    y dir=both,
    y explicit
] table [
    x=length,
    y=weighted_rate_avg,
    y error=weighted_rate_std
] {GTfix_del.dat};

\end{groupplot}
\end{tikzpicture}
\captionsetup{labelfont=bf}
\caption{Weighted error rates (according to the sequence read distribution) in percent and their standard deviations as a function of run-length $r$. }
\label{fig7}
\end{figure}
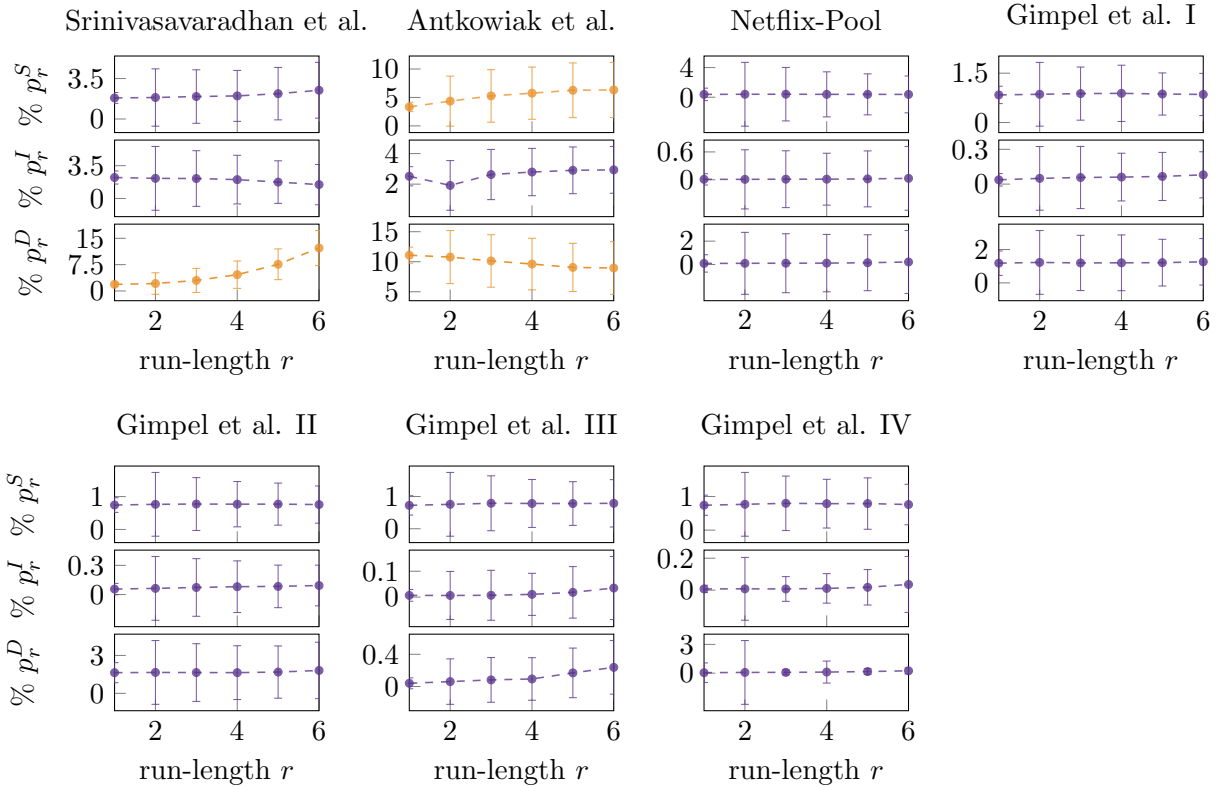

\begin{figure}[h!]
\centering
\begin{tikzpicture}
\begin{groupplot}[
    group style={
        group size=4 by 6,  
        vertical sep=2.2cm,   
        horizontal sep=1.2cm  
    },
    error bars/y dir=both,
    error bars/y explicit,
    width=4.3cm,
    height=2.6cm,
    scaled y ticks = false, 
    tick pos=left, 
    tick pos=bottom, 
    xmin=0.35, 
    xmax=0.65
]

\nextgroupplot[ title={Srinivasavaradhan et al.},ylabel={\% $p_{w}^S$}, xtick={0.4, 0.5, 0.6}, xticklabels={$ $, $ $, $ $}, ytick={0.015, 0.025}, yticklabels={$1.5$, $2.5$}, ymax=0.028]
\addplot [
    mark=o,
    mark size=0.1,
    opacity=0.8,
    line width=0.6pt,
    color=Peach,
    dashed,
    forget plot
] table [
    x=GC,
    y=avg_r_substitutions
] {groupedM.dat};
\addplot [
    only marks,
    opacity=0.8,
    color=Peach,
    mark size=1.5,
    error bars/.cd,
    y dir=both,
    y explicit
] table [
    x=GC,
    y=avg_r_substitutions,
    y error=weighted_std_r_substitutions
] {groupedM.dat};

\nextgroupplot[ title={Antkowiak et al.}, xtick={0.4, 0.5, 0.6}, xticklabels={$ $, $ $, $ $}, ytick={0.03, 0.05}, yticklabels={$3$, $5$}, ymin=0.025 ]
\addplot [
    mark=o,
    mark size=0.1,
    opacity=0.8,
    line width=0.6pt,
    color=Peach,
    dashed,
    forget plot
] table [
    x=GC,
    y=avg_r_substitutions
] {groupedH.dat};
\addplot [
    only marks,
    opacity=0.8,
    color=Peach,
    mark size=1.5,
    error bars/.cd,
    y dir=both,
    y explicit
] table [
    x=GC,
    y=avg_r_substitutions,
    y error=weighted_std_r_substitutions
] {groupedH.dat};

\nextgroupplot[ title={Netflix-Pool}, xmin=0.35, xmax=0.65, xtick={0.4, 0.5, 0.6}, xticklabels={$ $, $ $, $ $}, ytick={0, 0.02}, yticklabels={$0$, $2$}, ymax=0.025 ]
\addplot [
    mark=o,
    mark size=0.1,
    opacity=0.8,
    line width=0.6pt,
    color=Orchid,
    dashed,
    forget plot
] table [
    x=GC,
    y=avg_r_substitutions
] {grouped.dat};
\addplot [
    only marks,
    opacity=0.8,
    color=Orchid,
    mark size=1.5,
    error bars/.cd,
    y dir=both,
    y explicit
] table [
    x=GC,
    y=avg_r_substitutions,
    y error=weighted_std_r_substitutions
] {grouped.dat};

\nextgroupplot[title={Gimpel et al. \uppercase\expandafter{\romannumeral1}}, clip=true,  xmin=0.35, xmax=0.65, xtick={0.4, 0.5, 0.6}, xticklabels={$ $, $ $, $ $},ytick={0.005, 0.01}, yticklabels={$0.5$, $1$}, ymax=0.014 ]
\addplot [
    mark=o,
    mark size=0.1,
    opacity=0.8,
    line width=0.6pt,
    color=Orchid,
    dashed,
    forget plot
] table [
    x=GC,
    y=avg_r_substitutions
] {groupedG.dat};
\addplot [
    only marks,
    opacity=0.8,
    color=Orchid,
    mark size=1.5,
    error bars/.cd,
    y dir=both,
    y explicit
]table [
    x=GC,
    y=avg_r_substitutions,
    y error=weighted_std_r_substitutions
] {groupedG.dat};

\nextgroupplot[ yshift=2.1cm, ylabel={\% $p_{w}^I$}, xmin=0.35, xmax=0.65, xtick={0.4, 0.5, 0.6}, xticklabels={$ $, $ $, $ $}, ytick={0.015, 0.025}, yticklabels={$1.5$, $2.5$}, ymin=0.008, ymax=0.027]
\addplot [
    mark=o,
    mark size=0.1,
    opacity=0.8,
    line width=0.6pt,
    color=Peach,
    dashed,
    forget plot
] table [
    x=GC,
    y=avg_r_insertions
] {groupedM.dat};
\addplot [
    only marks,
    opacity=0.8,
    color=Peach,
    mark size=1.5,
    error bars/.cd,
    y dir=both,
    y explicit
]table [
    x=GC,
    y=avg_r_insertions,
    y error=weighted_std_r_insertions
] {groupedM.dat};

\nextgroupplot[ yshift=2.1cm, xmin=0.35, xmax=0.65, xtick={0.4, 0.5, 0.6}, xticklabels={$ $, $ $, $ $}, ytick={0.02, 0.03}, yticklabels={$2$, $3$}, ymax=0.032]
\addplot [
    mark=o,
    mark size=0.1,
    opacity=0.8,
    line width=0.6pt,
    color=Orchid,
    dashed,
    forget plot
] table [
    x=GC,
    y=avg_r_insertions
] {groupedH.dat};
\addplot [
    only marks,
    opacity=0.8,
    color=Orchid,
    mark size=1.5,
    error bars/.cd,
    y dir=both,
    y explicit
]table [
    x=GC,
    y=avg_r_insertions,
    y error=weighted_std_r_insertions
] {groupedH.dat};

\nextgroupplot[ yshift=2.1cm, xmin=0.35, xmax=0.65, xtick={0.4, 0.5, 0.6}, xticklabels={$ $, $ $, $ $}, ytick={0, 0.001}, yticklabels={$0$, $0.1$}, ymax=0.0015]
\addplot [
    mark=o,
    mark size=0.1,
    opacity=0.8,
    line width=0.6pt,
    color=Orchid,
    dashed,
    forget plot
]table [
    x=GC,
    y=avg_r_insertions
] {grouped.dat};
\addplot [
    only marks,
    opacity=0.8,
    color=Orchid,
    mark size=1.5,
    error bars/.cd,
    y dir=both,
    y explicit
] table [
    x=GC,
    y=avg_r_insertions,
    y error=weighted_std_r_insertions
] {grouped.dat};

\nextgroupplot[,yshift=2.1cm, xmin=0.35, xmax=0.65, xtick={0.4, 0.5, 0.6}, xticklabels={$ $, $ $, $ $}, ytick={0, 0.001}, yticklabels={$0$, $0.1$}, ymax=0.0015]
\addplot [
    mark=o,
    mark size=0.1,
    opacity=0.8,
    line width=0.6pt,
    color=Orchid,
    dashed,
    forget plot
] table [
    x=GC,
    y=avg_r_insertions
] {groupedG.dat};
\addplot [
    only marks,
    opacity=0.8,
    color=Orchid,
    mark size=1.5,
    error bars/.cd,
    y dir=both,
    y explicit
] table [
    x=GC,
    y=avg_r_insertions,
    y error=weighted_std_r_insertions
] {groupedG.dat};


\nextgroupplot[yshift=2.1cm, xlabel={GC content $w$},ylabel={\% $p_{w}^D$}, xmin=0.35, xmax=0.65, xtick={0.4, 0.5, 0.6}, ytick={0.015, 0.025}, yticklabels={$1.5$, $2.5$}, ymin=0.01]
\addplot [
    mark=o,
    mark size=0.1,
    opacity=1,
    line width=0.6pt,
    color=Peach,
    dashed,
    forget plot
] table [
    x=GC,
    y=avg_r_deletions
] {groupedM.dat};
\addplot [
    only marks,
    opacity=0.8,
    color=Peach,
    mark size=1.5,
    error bars/.cd,
    y dir=both,
    y explicit
] table [
    x=GC,
    y=avg_r_deletions,
    y error=weighted_std_r_deletions
] {groupedM.dat};

\nextgroupplot[yshift=2.1cm, xlabel={GC content $w$}, xmin=0.35, xmax=0.65, xtick={0.4, 0.5, 0.6},  ytick={0.1, 0.12}, yticklabels={$10$, $12$}]
\addplot [
    mark=o,
    mark size=0.1,
    opacity=1,
    line width=0.6pt,
    color=Orchid,
    dashed,
    forget plot
] table [
    x=GC,
    y=avg_r_deletions
] {groupedH.dat};
\addplot [
    only marks,
    opacity=0.8,
    color=Orchid,
    mark size=1.5,
    error bars/.cd,
    y dir=both,
    y explicit
] table [
    x=GC,
    y=avg_r_deletions,
    y error=weighted_std_r_deletions
] {groupedH.dat};

\nextgroupplot[yshift=2.1cm, xmin=0.35, xmax=0.65, xtick={0.4, 0.5, 0.6},  xlabel={GC content $w$}, ytick={0, 0.01}, yticklabels={$0$, $1$}]
\addplot [
    mark=o,
    mark size=0.1,
    opacity=0.8,
    line width=0.6pt,
    color=Orchid,
    dashed,
    forget plot
] table [
    x=GC,
    y=avg_r_deletions
] {grouped.dat};
\addplot [
    only marks,
    opacity=0.8,
    color=Orchid,
    mark size=1.5,
    error bars/.cd,
    y dir=both,
    y explicit
] table [
    x=GC,
    y=avg_r_deletions,
    y error=weighted_std_r_deletions
] {grouped.dat};

\nextgroupplot[yshift=2.1cm, xmin=0.35, xmax=0.65, xtick={0.4, 0.5, 0.6},  xlabel={GC content $w$}, ytick={0.01, 0.022}, yticklabels={$1$, $2$}]
\addplot [
    mark=o,
    mark size=0.1,
    opacity=0.8,
    line width=0.6pt,
    color=Orchid,
    dashed,
    forget plot
] table [
    x=GC,
    y=avg_r_deletions
] {groupedG.dat};
\addplot [
    only marks,
    opacity=0.8,
    color=Orchid,
    mark size=1.5,
    error bars/.cd,
    y dir=both,
    y explicit
] table [
    x=GC,
    y=avg_r_deletions,
    y error=weighted_std_r_deletions
] {groupedG.dat};

\nextgroupplot[title={Gimpel et al. \uppercase\expandafter{\romannumeral2}},ylabel={\% $p_{w}^S$}, xmin=0.35, xmax=0.65, xtick={0.4, 0.5, 0.6}, xticklabels={$ $, $ $, $ $}, ytick={0.006, 0.009}, yticklabels={$0.6$, $0.9$}]
\addplot [
    mark=o,
    mark size=0.1,
    opacity=0.8,
    line width=0.6pt,
    color=Orchid,
    dashed,
    forget plot
] table [
    x=GC,
    y=avg_r_substitutions
] {groupedGf.dat};
\addplot [
    only marks,
    opacity=0.8,
    color=Orchid,
    mark size=1.5,
    error bars/.cd,
    y dir=both,
    y explicit
] table [
    x=GC,
    y=avg_r_substitutions,
    y error=weighted_std_r_substitutions
] {groupedGf.dat};

\nextgroupplot[title={Gimpel et al. \uppercase\expandafter{\romannumeral3}}, ytick={0.005, 0.015},xmin=0.35, xmax=0.65, xtick={0.4, 0.5, 0.6}, xticklabels={$ $, $ $, $ $}, ytick={0.005, 0.01}, yticklabels={$0.5$, $1$}]
\addplot [
    mark=o,
    mark size=0.1,
    opacity=0.8,
    line width=0.6pt,
    color=Orchid,
    dashed,
    forget plot
] table [
    x=GC,
    y=avg_r_substitutions
] {groupedT.dat};
\addplot [
    only marks,
    opacity=0.8,
    color=Orchid,
    mark size=1.5,
    error bars/.cd,
    y dir=both,
    y explicit
] table [
    x=GC,
    y=avg_r_substitutions,
    y error=weighted_std_r_substitutions
] {groupedT.dat};

\nextgroupplot[title={Gimpel et al. \uppercase\expandafter{\romannumeral4}}, clip=true, xmin=0.35, xmax=0.65, xtick={0.4, 0.5, 0.6}, xticklabels={$ $, $ $, $ $}, ytick={0.005, 0.01}, yticklabels={$0.5$, $1$}, ymin=0.003, ymax=0.013]
\addplot [
    mark=o,
    mark size=0.1,
    opacity=0.8,
    line width=0.6pt,
    color=Orchid,
    dashed,
    forget plot
] table [
    x=GC,
    y=avg_r_substitutions
] {groupedTf.dat};
\addplot [
    only marks,
    opacity=0.8,
    color=Orchid,
    mark size=1.5,
    error bars/.cd,
    y dir=both,
    y explicit
] table [
    x=GC,
    y=avg_r_substitutions,
    y error=weighted_std_r_substitutions
] {groupedTf.dat};

\nextgroupplot[hide axis]

\nextgroupplot[yshift=2.1cm, ylabel={\% $p_{w}^I$}, xtick={0.4, 0.5, 0.6}, xticklabels={$ $, $ $, $ $}, ytick={0, 0.001}, yticklabels={$0$, $0.1$},ymin=-0.0005, ymax=0.0015]
\addplot [
    mark=o,
    mark size=0.1,
    opacity=0.8,
    line width=0.6pt,
    color=Orchid,
    dashed,
    forget plot
] table [
    x=GC,
    y=avg_r_insertions
] {groupedGf.dat};
\addplot [
    only marks,
    opacity=0.8,
    color=Orchid,
    mark size=1.5,
    error bars/.cd,
    y dir=both,
    y explicit
] table [
    x=GC,
    y=avg_r_insertions,
    y error=weighted_std_r_insertions
] {groupedGf.dat};

\nextgroupplot[ yshift=2.1cm, xmin=0.35, xmax=0.65, xtick={0.4, 0.5, 0.6}, xticklabels={$ $, $ $, $ $}, , ytick={0, 0.0003}, yticklabels={$0$, $0.03$}]
\addplot [
    mark=o,
    mark size=0.1,
    opacity=0.8,
    line width=0.6pt,
    color=Orchid,
    dashed,
    forget plot
] table [
    x=GC,
    y=avg_r_insertions
] {groupedT.dat};
\addplot [
    only marks,
    opacity=0.8,
    color=Orchid,
    mark size=1.5,
    error bars/.cd,
    y dir=both,
    y explicit
] table [
    x=GC,
    y=avg_r_insertions,
    y error=weighted_std_r_insertions
] {groupedT.dat};

\nextgroupplot[ yshift=2.1cm, xmin=0.35, xmax=0.65, xtick={0.4, 0.5, 0.6}, xticklabels={$ $, $ $, $ $}, ytick={0, 0.0003}, yticklabels={$0$, $0.03$}, ymax=0.00038]
\addplot [
    mark=o,
    mark size=0.1,
    opacity=0.8,
    line width=0.6pt,
    color=Orchid,
    dashed,
    forget plot
] table [
    x=GC,
    y=avg_r_insertions
] {groupedTf.dat};
\addplot [
    only marks,
    opacity=0.8,
    color=Orchid,
    mark size=1.5,
    error bars/.cd,
    y dir=both,
    y explicit
] table [
    x=GC,
    y=avg_r_insertions,
    y error=weighted_std_r_insertions
] {groupedTf.dat};

\nextgroupplot[hide axis]

\nextgroupplot[yshift=2.1cm, ylabel={\% $p_{w}^D$}, xmin=0.35, xmax=0.65, xtick={0.4, 0.5, 0.6}, xlabel={GC content $w$}, , ytick={0.015, 0.025}, yticklabels={$1.5$, $2.5$}, ymax=0.029]
\addplot [
    mark=o,
    mark size=0.1,
    opacity=0.8,
    line width=0.6pt,
    color=Orchid,
    dashed,
    forget plot
] table [
    x=GC,
    y=avg_r_deletions,
] {groupedGf.dat};
\addplot [
    only marks,
    opacity=0.8,
    color=Orchid,
    mark size=1.5,
    error bars/.cd,
    y dir=both,
    y explicit
] table [
    x=GC,
    y=avg_r_deletions,
    y error=weighted_std_r_deletions
] {groupedGf.dat};

\nextgroupplot[yshift=2.1cm,  xmin=0.35, xmax=0.65, xtick={0.4, 0.5, 0.6},  xlabel={GC content $w$}, , ytick={0, 0.001}, yticklabels={$0$, $0.1$}]
\addplot [
    mark=o,
    mark size=0.1,
    opacity=0.8,
    line width=0.6pt,
    color=Orchid,
    dashed,
    forget plot
] table [
    x=GC,
    y=avg_r_deletions
] {groupedT.dat};
\addplot [
    only marks,
    opacity=0.8,
    color=Orchid,
    mark size=1.5,
    error bars/.cd,
    y dir=both,
    y explicit
] table [
    x=GC,
    y=avg_r_deletions,
    y error=weighted_std_r_deletions
] {groupedT.dat};

\nextgroupplot[, yshift=2.1cm,  xmin=0.35, xmax=0.65, xtick={0.4, 0.5, 0.6},  xlabel={GC content $w$}, ytick={0, 0.001}, yticklabels={$0$, $0.1$}]
\addplot [
    mark=o,
    mark size=0.1,
    opacity=0.7,
    line width=0.6pt,
    color=Orchid,
    dashed,
    forget plot
] table [
    x=GC,
    y=avg_r_deletions
] {groupedTf.dat};
\addplot [
    only marks,
    opacity=0.8,
    color=Orchid,
    mark size=1.5,
    error bars/.cd,
    y dir=both,
    y explicit
] table [
    x=GC,
    y=avg_r_deletions,
    y error=weighted_std_r_deletions
] {groupedTf.dat};

\end{groupplot}
\end{tikzpicture}
\captionsetup{labelfont=bf}
\caption{Weighted error rates (according to the sequence read distribution) in percent and their standard deviations as a function of GC content $w$. } 
\label{fig8}
\end{figure}
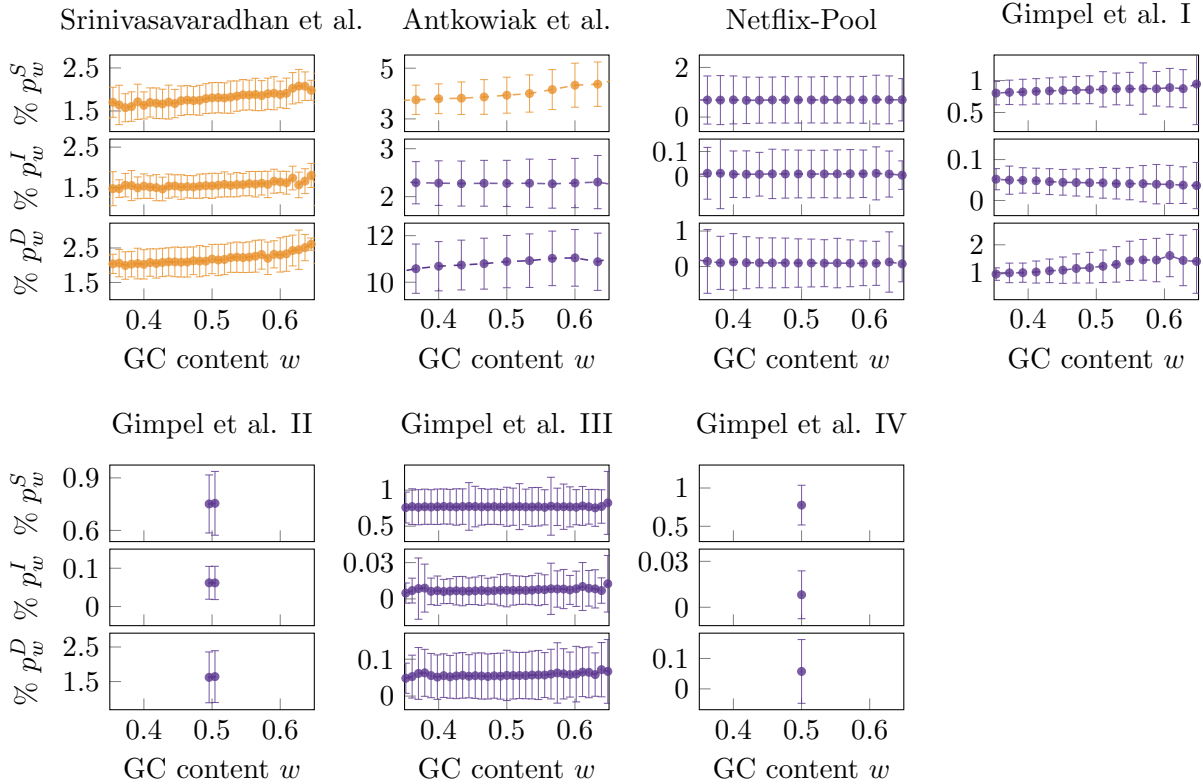

Figure~\ref{fig8} summarizes the error rates as a function of sequence GC content for all datasets considered. In \citet{gimpel_digital_2023} and the Netflix-Pool, we estimate no correlation between GC content and all error rates. In the dataset by \citet{srinivasavaradhan_trellis_2021}, we find a minor linear increase in all error rates with GC content. However, the increase is within the one percentage point region, and the increase in the substitution rate is far from the error regimes established in Section~\ref{lbc}, in which constrained coding for the GC content is efficient. Similarly, the substitution rate increase observed in the dataset by \citet{antkowiak_low_2020} lies outside the error regimes in which constrained coding for the GC content is efficient.

There are several reasons that can explain the difference between our experimental results and those in the literature \citep{ross_characterizing_2013, bragg_shining_2013, bohlin_estimation_2019, stoler_sequencing_2021}, which suggest that homopolymers and GC content imbalances increase error rates. The key difference is that most papers on constrained coding for DNA data storage cite studies estimating error rates for DNA sequences stored in vivo. For example, the frequently cited study by \citet{ross_characterizing_2013} estimates error rates for different sequencing technologies using human and bacterial DNA probes. A similar approach is taken in the more recent error analysis by \citet{laehnemann_denoising_2016}.

The occurrence and length of runs can vary among organisms and genomic regions, with parts of the genome exhibiting frequent long runs. Next-generation sequencing technologies use "sequencing by synthesis" to read the DNA strands and chemicals to detect the incorporation of a nucleotide into the growing polymer chain \citep{niedringhaus_landscape_2011}. The signals emitted by the chemicals are amplified in long runs and can accumulate if not completely removed after a sequencing cycle, leading to 'post-homopolymer substitutions' \citep{stoler_sequencing_2021}. The increase in substitution, insertion, and deletion rates may be due to the sequencer not being able to correctly identify long run-lengths from the amplified signal, as run-length and signal intensity do not necessarily match perfectly \citep{pereira_bioinformatics_2020, laehnemann_denoising_2016}. However, in DNA storage, where sequences are randomized, long runs occur infrequently, and their impact on the average error rates is minimal.

Similarly, GC content can vary significantly across different regions of the human genome and among organisms, potentially leading to higher error rates. For example, the \textit{Plasmodium falciparum} and \textit{Rhodobacter sphaeroides} bacteria have a mean GC content of 19\% and 69\%, respectively \citep{ross_characterizing_2013}. However, in randomized DNA sequences, the average GC content is balanced.

Our empirical analysis (see Appendix~\ref{App D}, Figure~\ref{App14}) and studies \citep{benjamini_summarizing_2012, dohm_substantial_2008} find that GC content is correlated with read coverage, and more sequencing is required for DNA sequences with unbalanced GC content. This may be because sequences with extreme GC content are amplified less efficiently during PCR amplification \citep{kozarewa_amplification-free_2009}. However, our theoretical results do not provide information on the cost of additional sequencing compared to the cost of avoiding error-prone sequences. No direct comparison can be made here, but the cost of additional sequencing must be weighed against the cost of more synthesis.

\section{Conclusion}
Our results suggest that in most current DNA storage systems, embracing substitution errors is more efficient than avoiding them through constrained coding.

However, our channel models have certain limitations over practical DNA storage systems. First, both the run-length varying and the GC content channel do not account for asymmetric error probabilities observed in practice. For example, \citet{heckel_characterization_2019} finds that in their experiments, substitutions from $C$ to $T$ and $G$ to $A$ are the most frequent. Second, our analysis is limited to substitution errors and does not address deletions, insertions, or molecular impairments such as strand breakage because their channel capacities and, hence, achievable code rates are unknown to date.

Therefore, while this study highlights the use of unconstrained coding, constrained coding may still prove useful in different systems, and future DNA storage technologies or channels that include deletions, insertions and molecular impairments may very well lead to such systems.

\section*{Acknowledgment}
The authors thank Maria Abu-Sini, Antonia Wachter-Zeh and Eitan Yaakobi for helpful discussions and feedback.

The research leading to these results received funding from the European Union under the Horizon 2020 Program, FET-Open: DNA-FAIRYLIGHTS, Grant Agreement No. 964995 and FET-Open: DiDAX, Grant Agreement No. 101115134. Views and opinions expressed are however those of the author(s) only and do not necessarily reflect those of the European Union or the European Research Council Executive Agency. Neither the European Union nor the granting authority can be held responsible for them.

\printbibliography
\appendix
\label{appendix}
\section{Proof of Theorem~\ref{th1}}
\label{pf1}

The proof follows the random coding argument of the achievability part of Shannon’s Channel Coding Theorem \citep{shannon_mathematical_1948}. The difference to Shannon's original proof is that we introduce a random variable $R$ which represents the run-length of a given nucleotide. This allows treating the channel output as conditionally independent of other channel inputs and outputs by conditioning on both the transmitted nucleotide and its run-length.

We first prove that the code rate $R_c$ is achievable for $m$-constrained coding. We generate an \(m\)-constrained code, \(\mathcal{C}_m\), of length $n$ with rate \(R_c\) as follows. The encoding function \(f_m\) maps each message index \(w\) sequentially from \(1\) to \(2^{R_c n}\) to a codeword $\mathbf{X}$, where message index \(1\) corresponds to the first generated codeword, message index \(2\) to the second, and so on, up to \(2^{R_c n}\). 

Each codeword \(\mathbf{X} = X_1 \ldots X_n\) consists of nucleotides \(X_i\), where each \(X_i\) is generated by a Markov chain of order \(m\) to satisfy the maximum run-length constraint:
\begin{equation}
\label{X}
    \Pr(f_m(w) = \mathbf{x}) = \prod_{i=1}^n \Pr(X_i = x_i \mid x_{i-m} \ldots x_{i-1}),
\end{equation}
where
\begin{equation}
\Pr(X_i = x_i \mid x_{i-m} \ldots x_{i-1}) = 
\begin{cases} 
0 & \text{if } x_i = x_{i-1} \text{ and } x_{i-m} = x_{i-m+1} = \ldots = x_{i-1}, \\
\frac{1}{3} & \text{if } x_i \neq x_{i-1} \text{ and } x_{i-m} = x_{i-m+1} = \ldots = x_{i-1}, \\
\frac{1}{4} & \text{otherwise}.
\end{cases}
\end{equation}

As an example, with \(m=1\), where consecutive nucleotides cannot be of the same type (i.e., \(AA\) cannot occur), the transition matrix \(\mathbf{B}\) is:
\begin{equation}
\label{B}
   \mathbf{B} = \begin{bmatrix}
0 & \frac{1}{3} & \frac{1}{3} & \frac{1}{3} \\
\frac{1}{3} & 0 & \frac{1}{3} & \frac{1}{3} \\
\frac{1}{3} & \frac{1}{3} & 0 & \frac{1}{3} \\
\frac{1}{3} & \frac{1}{3} & \frac{1}{3} & 0 \\
\end{bmatrix}.  
\end{equation}
The probability of generating the entire code \(\mathcal{C}_m\) with rate \(R_c\) is:
\[
\Pr(\mathcal{C}_m) = \prod_{w=1}^{2^{R_c n}} \Pr(f_m(w) = \mathbf{x}) = \prod_{w=1}^{2^{R_c n}} \prod_{i=1}^n \Pr(X_i = x_i \mid x_{i-m} \ldots x_{i-1}).
\]
Both the code \(\mathcal{C}_m\) and the channel characteristics, defined in Definition~\ref{RLV Channel}, are known to the receiver and the transmitter. Transmission occurs by selecting a message index \( W \) uniformly at random from the set of all message indices \(\{1, 2, \ldots, 2^{R_c n}\}\), where each index is chosen with equal probability \(1/2^{R_c n}\). The codeword \( f_m(W=w) = \mathbf{x} \) is transmitted, and the receiver receives the output sequence \(\mathbf{Y}\), distributed according to:
\begin{equation}
\label{Y}
\Pr(\mathbf{Y} = \mathbf{y} \mid \mathbf{x}) = \prod_{i=1}^n \Pr(Y_i = y_i \mid x_i, r_i),
\end{equation}
where \( r_i \) is the run-length of nucleotide \( x_i \) which is determined by $\vx$.

To determine which codeword was transmitted, the receiver uses joint typicality decoding. Originally, joint typicality decoding is defined for independent and identically distributed nucleotides in a memoryless channel. We adjust this definition to account for dependencies in codeword generation and that the substitution probability depends on the run-length, which is determined by the adjacent input nucleotides.

Define \(P^{\mathbf{X}}=\{X_i\}_{i=1}^n\) as the stationary and ergodic stochastic process that generates the codewords \(\mathbf{X}\) according to Equation \eqref{X}, and \(P^{\mathbf{Y}}=\{Y_i\}_{i=1}^n\) as the stationary and ergodic stochastic process that generates the output sequences \(\mathbf{Y}\) according to Equation \eqref{Y}.

A stochastic process \( P^{\mathbf{Z}} \) is defined as a sequence of random variables \(\{Z_i\}_{i=1}^n\), where each random variable \(Z_i\) takes values in alphabet \(\mathcal{Z}\). The process is characterized by its joint probability distribution:
\[
\Pr(Z_1 = z_1, Z_2 = z_2, \ldots, Z_n = z_n),
\]
for each \(n = 1, 2, \ldots\). A stochastic process is said to be stationary if its statistical properties do not change over time and is said to be ergodic if time averages converge to ensemble averages. The entropy rate of a stationary and ergodic stochastic process is defined as:
\[
H\left(P^{\mathbf{Z}}\right) = \lim_{n \to \infty} \frac{1}{n} H(Z_1, Z_2, \ldots, Z_n),
\]
where the limit exists \citep{cover_elements_2012}.

Let us further define the set \(\mathcal{T}_{\epsilon}\) of \(\epsilon\)-joint typical sequences \((\mathbf{x}, \mathbf{y})\) as follows. 

\begin{definition}
\textbf{\emph{Adjusted from Section 7.6. \cite{cover_elements_2012}}.}
\label{T}
Define the set $\mathcal{T}_{\epsilon}$ of jointly typical sequences with respect to the joint distribution $\Pr(\mathbf{X}=\mathbf{x}, \mathbf{Y}=\mathbf{y})$ as the set of sequences $(\mathbf{x}, \mathbf{y})$ whose empirical entropies are $\epsilon$-close to the true entropies:
\begin{align*}
\mathcal{T}_{\epsilon} = \Bigl\{ &(\mathbf{x}, \mathbf{y}) \in \{A,C,G,T\}^n \times \{A,C,G,T\}^n : \\
&\Bigl| -\frac{1}{n} \log \Pr(\mathbf{X}=\mathbf{x}) - H\left(P^{\mathbf{X}}\right) \Bigr| < \epsilon, \\
&\Bigl| -\frac{1}{n} \log \Pr(\mathbf{Y}=\mathbf{y}) - H\left(P^{\mathbf{Y}}\right) \Bigr| < \epsilon, \\
&\Bigl| -\frac{1}{n} \log \Pr(\mathbf{X}=\mathbf{x}, \mathbf{Y}=\mathbf{y}) - H\left(P^{\mathbf{X}}, P^{\mathbf{Y}}\right) \Bigr| < \epsilon \Bigr\},
\end{align*}
where \(P^{\mathbf{X}}\) and \(P^{\mathbf{Y}}\) are defined as above.
\end{definition}

We are now ready to define the decoding rule by which the receiver declares message indices for received sequences. The receiver declares the message index \( \hat{w} \) for the received sequence \( \mathbf{y} \) if it is jointly typical with the transmitted codeword \( f_m(w) = \mathbf{x} \) and not with any other codeword \( f_m(w') \) for \( w' \neq \hat{w} \). If no such \( \hat{w} \) exists, an error is declared.

Thus, the receiver makes a decoding error if the received sequence \( \mathbf{y} \) is either not jointly typical with \( \mathbf{x} \) or is jointly typical with a codeword of a different index. The probability of these events is given by the joint asymptotic equipartition property stated in Theorem \ref{AEP}. The joint asymptotic equipartition property holds for stationary and ergodic Markov processes due to the Shannon-McMillan-Breiman theorem. The Shannon-McMillan-Breiman theorem states that if \( H(P^{\mathbf{Z}}) \) is the entropy rate of a finite-valued stationary ergodic process \(P^{\mathbf{Z}}=\{Z_i\}\), then \( -\frac{1}{n} \log \Pr(Z_1, \ldots, Z_n= z_1, \ldots, z_n) \) converges to \( H(P^{\mathbf{Z}}) \) with probability $1$ \citep{cover_elements_2012}. 

\begin{theorem}
\textbf{\emph{Adjusted from Theorem 7.6.1. \citep{cover_elements_2012}.}}
\label{AEP}
For sequences \((\mathbf{x}, \mathbf{y})\) of length \(n\) that are drawn independently and identically according to the distribution \(\Pr(\mathbf{X}=\mathbf{x}, \mathbf{Y}=\mathbf{y})\):
\begin{enumerate}
    \item \(\Pr((\mathbf{X}=\mathbf{x}, \mathbf{Y}=\mathbf{y}) \in  \mathcal{T}_{\epsilon}) \to 1\) as \(n \to \infty\).
    \item \(|\mathcal{T}_{\epsilon}| \leq 2^{n(H\left(P^{\mathbf{X}}, P^{\mathbf{Y}}\right) + \epsilon)}\).
    \item If \((\tilde{\mathbf{X}}, \tilde{\mathbf{Y}}) \sim \Pr(\mathbf{X}=\mathbf{x})\Pr(\mathbf{Y}=\mathbf{y})\) [i.e., \(\tilde{\mathbf{X}}\) and \(\tilde{\mathbf{Y}}\) are independent but have the same marginals as \(Pr(\mathbf{X}=\mathbf{x}, \mathbf{Y}=\mathbf{y})\)], then
    \[
    \Pr((\tilde{\mathbf{X}}, \tilde{\mathbf{Y}}) \in \mathcal{T}_{\epsilon}) \leq 2^{-n\left(I\left(P^{\mathbf{X}}; P^{\mathbf{Y}})\right) - 3\epsilon\right)}.
    \]
\end{enumerate}
\end{theorem}

Following Shannon's random coding argument, we analyze the probability of decoding error over the random choice of a codebook to show that there exists at least one codebook with a probability of decoding error approaching zero as the sequence length \(n \to \infty\). The average probability of decoding error, averaged over all codebooks and codewords, can be bounded as follows:

\begin{align*}
&\sum_{\mathcal{C}_m} \Pr(\mathcal{C}_m) \sum_{w=1}^{2^{R_cn}} \Pr(f_m(w) \neq f_m(\hat{w})) \\
\leq &\sum_{\mathcal{C}_m} \Pr(\mathcal{C}_m) \left( \Pr((f_m(1), \mathbf{Y}) \not\in T^\epsilon) + \sum_{w=2}^{2^{R_cn}} \Pr((f_m(w), \mathbf{Y}) \in T^\epsilon) \right) \\
\leq &\sum_{\mathcal{C}_m} \Pr(\mathcal{C}_m) \left( \epsilon + \sum_{w=2}^{2^{R_cn}} 2^{-n(I(P^{\mathbf{X}};P^{\mathbf{Y}}) - 3\epsilon)} \right) \\
= &\epsilon + (2^{R_cn} - 1) 2^{-n(I(P^{\mathbf{X}};P^{\mathbf{Y}}) - 3\epsilon)}.
\end{align*}
In the first inequality, we use the union bound and, without loss of generality, assume message index \(w = 1\) was transmitted, given the constant probability of decoding error for all message indices due to the random code generation. The receiver makes an error if the received sequence \(\mathbf{Y}\) is not jointly typical with \(f_m(1)\), or if it is jointly typical with any other codeword \(f_m(i)\) for \(i = 2, \dots, 2^{R_{cn}}\). By the joint asymptotic equipartition property, the probability of the former approaches $1$, while the probability of the latter is less than \(2^{-n(I(P^{\mathbf{X}};P^{\mathbf{Y}}) - 3\epsilon)}\) as \(n \to \infty\).

For code rates \(R_c \leq I(P^{\mathbf{X}}; P^{\mathbf{Y}}) - 3\epsilon\), we can choose \(\epsilon\) and \(n\) such that the average probability of decoding error is less than \(2\epsilon\), making it arbitrarily small as \(n \to \infty\). To further show that the maximum probability of decoding error (over all codewords) approaches zero, we consider the following standard argument. Since the average probability of decoding error over all codes and their codewords tends to zero, there must be at least one code, \(\mathcal{C}_m^*\), where this holds. Further, for \(\mathcal{C}_m^*\), the maximum probability of decoding error over all codewords goes to zero, if we discard the worst half of the codewords, resulting in a rate loss of \(1/n\), which is negligible since \(n \to \infty\).

So far, we have shown that all rates \(R_c \leq I\left(P^{\mathbf{X}}; P^{\mathbf{Y}}\right)\) are achievable for \(m\)-constrained coding. Next, we compute the mutual information \(I\left(P^{\mathbf{X}};P^{\mathbf{Y}}\right)\) between the input process \(P^{\mathbf{X}}\) and the output process \(P^{\mathbf{Y}}\):
\begin{align*}
    I\left(P^{\mathbf{X}}; P^{\mathbf{Y}}\right) &{=} H\left(P^{\mathbf{Y}}\right) - H\left(P^{\mathbf{Y}}|P^{\mathbf{X}}\right) \\ 
    &\stackrel{(a)}{=} H\left(P^{\mathbf{Y}}\right) - \lim_{n \to \infty} \frac{1}{n}H\left(\mathbf{Y}|\mathbf{X}\right) \\
    &\stackrel{(b)}{=} H\left(P^{\mathbf{Y}}\right) - \lim_{n \to \infty} \frac{1}{n}H\left(\mathbf{Y}|\mathbf{X}, \mathbf{R}\right) \\
    &\stackrel{(c)}{=} H\left(P^{\mathbf{Y}}\right) - \lim_{n \to \infty} \frac{1}{n} \sum_{i=1}^n H\left(Y_i | X_i, R_i\right) \\
    &\stackrel{(d)}{=} H\left(P^{\mathbf{Y}}\right) - \lim_{n \to \infty} \frac{1}{n} \sum_{i=1}^n \sum_{r=1}^m \Pr(R_i=r) H\left(Y_i | X_i, r\right),
\end{align*}
where in step (a), we use the definition of the entropy of a stochastic process. In step (b), we use the fact that conditioning on \(\mathbf{R}\) does not change the entropy, as the sequence of run-lengths \(\mathbf{R}=R_1, \ldots, R_n\) is completely determined by the codeword \(\mathbf{X}=X_1, \ldots, X_n\). In step (c), we use the conditional independence of output nucleotide \(Y_i\) given input nucleotide \(X_i\) and its run-length \(R_i\). The probability \(\Pr(R_i=r)\) is the probability that a nucleotide \(X_i\) generated according to stochastic process \(P^{\mathbf{X}}\) occurs in a run of length \(r\). We derive an expression for the asymptotic distribution of \(\Pr(R_i=r)\) in Lemma~\ref{Lemma:2}.

\begin{lemma}
\textbf{\emph{Asymptotic distribution of run-lengths:}}
\label{Lemma:2}
Let \( \mathbf{X} \) be chosen uniformly at random from the set of \(m\)-constrained sequences \(\mathcal{A}_m\). 
For any nucleotide \( X_i \) at position \( i \), the probability that \( X_i \) is part of a run of length exactly \( r \) converges in distribution to:
\[
\Pr(R_i=r) \xrightarrow{d} q_m(r)=\frac{r\left(\frac{1}{4}\right)^{r-1} \left(\frac{3}{4}\right)^2}{\sum_{s=1}^m s\left(\frac{1}{4}\right)^{s-1}  \left(\frac{3}{4}\right)^2} \text{ as } n \to \infty.
\]
\end{lemma}

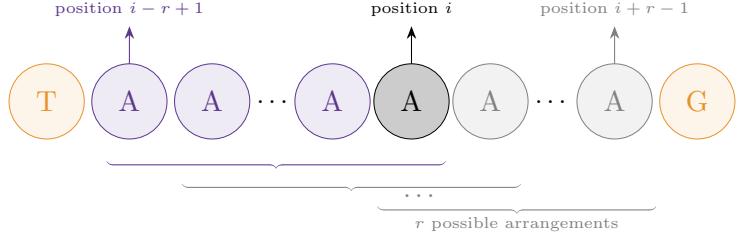
\begin{figure}[t]
    \centering
    \begin{tikzpicture}[>=Stealth]
            \draw[very thin, draw=Peach, fill=Peach!10, text=Peach] (-0.8,3) circle[radius=0.5] node at (-0.8,3) {T};
            \draw[ very thin,  draw=Orchid, fill=Orchid!10, text=Orchid] (0.3,3) circle[radius=0.5] node at (0.3,3) {A};
            \draw[draw=Orchid, fill=Orchid!10, text=Orchid, very thin] (1.4,3) circle[radius=0.5] node at (1.4,3) {A};
        \draw (2.2,3) node {\small{$\ldots$}};
        \draw[draw=Orchid, fill=Orchid!10, text=Orchid, very thin] (3,3) circle[radius=0.5] node at (3,3) {A};
        \draw[very thin, draw=black, fill=gray!40, text=black] (4.05,3) circle[radius=0.5] node at (4.05,3) {A};
        \draw[very thin, draw=gray, fill=gray!10, text=gray] (5.1,3) circle[radius=0.5] node at (5.1,3) {A};
        \draw (5.9,3) node {\small{$\ldots$}};
            \draw[very thin, draw=gray, fill=gray!10, text=gray] (6.75,3) circle[radius=0.5]node at (6.75,3) {A};
            \draw[draw=Peach, fill=Peach!10, text=Peach, very thin] (7.85,3) circle[radius=0.5]node at (7.85,3) {G};


\draw [pen colour=Orchid,decorate,decoration = {calligraphic brace}, ] (4.5,2.2)--(0,2.2) ;

\draw [pen colour=gray,decorate,decoration = {calligraphic brace}, ] (5.5,1.9)--(1,1.9) node[below , pos=0.3, color=gray]{\small{$\ldots$}};

\draw [pen colour=gray,decorate,decoration = {calligraphic brace}, ] (7.3,1.6)--(3.6,1.6)node[below , pos=0.5, color=gray]{\tiny{$r$ possible arrangements}} ;

 \draw[->, color=black] (4.05,3.5) --((4.05,4) node[above]{\tiny{position $i$}};
 \draw[->, color=Orchid] (0.3,3.5) --((0.3,4)node[above]{\tiny{position $i-r+1$}};
  \draw[->, color=gray] (6.75,3.5) --((6.75,4)node[above]{\tiny{position $i+r-1$}};
    \end{tikzpicture}
    \captionsetup{labelfont=bf}
    \caption{ Visual illustration of Lemma~\ref{Lemma:2}. There are $r$ possible locations for nucleotide $X_i$ at position $r-1<i<n-r+1$ to occur in the run of length $r$. }
    \label{fig9}
\end{figure}

\begin{proof}
First, consider the probability \(\Pr(R_i=r)\) when \(\mathbf{X}\) is chosen uniformly at random from the set of all possible sequences \(\{A,C,G,T\}^n\), where each nucleotide occurs with equal probability of \(1/4\). The probability of observing \(r-1\) consecutive identical nucleotides is \(\left(1/4\right)^{r-1}\). To form a run of length exactly \(r\), the run must be preceded and followed by a different nucleotide, occurring with a probability of \(\left(3/4\right)^2\). Within the run, nucleotide \(X_i\), where \(r-1 < i < n-r+1\), can be in any of the \(r\) positions. As the sequence length \(n\) increases, the effect of edge positions, where \(X_i\) has fewer than \(r\) possible positions, becomes negligible. Thus, as \(n\) approaches infinity, \(\Pr(R_i=r)\) for \(\mathbf{X} \in \{A,C,G,T\}^n\) converges in distribution to \(q(r)=r\left(1/4\right)^{r-1}\left(3/4\right)^2\). Figure~\ref{fig9} illustrates this. For sequences \(\mathbf{X} \in \mathcal{A}_m\), no nucleotide \(X_i\) can be part of a run longer than \(m\). Therefore, for run-lengths \(r \leq m\), we adjust the probabilities to account for the constrained space of possible sequences. The normalization factor is \(\Pr(R_i=r \leq m) = \sum_{r=1}^m \Pr(R_i=r)\), and the probability \(\Pr(R_i=r)\) converges in distribution to \(q_m(r)=q(r)/ \sum_{s=1}^m q(s)\).
\end{proof}


Using Lemma~\ref{Lemma:2}, we further have:
\begin{align*}
    I\left(P^{\mathbf{X}}; P^{\mathbf{Y}}\right) &{=} H\left(P^{\mathbf{Y}}\right) - \lim_{n \to \infty} \frac{1}{n} \sum_{i=1}^n \sum_{r=1}^n \Pr(R_i=r) H(Y_i | X_i, r) \\
    &{=} H\left(P^{\mathbf{Y}}\right) - \sum_{r=1}^n q(r) H(p_r),
\end{align*}
where \(H(p_r)\) is defined in Theorem~\ref{th1} and is the entropy of a quaternary random variable that retains its state with probability \(1-p_r\) and substitutes to one of the other three states with probability \(p_r/3\). This concludes the proof for the achievable code rate \(R_c\) in Equation~\eqref{R_c} of Theorem~\ref{th1}.

The proof for the achievable code rate \(R_u\) in Equation~\eqref{R_u} for unconstrained coding follows directly by generating codewords with no run-length constraint according to a uniform distribution over the nucleotides, where each nucleotide occurs with an equal probability of \(1/4\). Therefore, we omit the proof.

For unconstrained coding, the distribution over the output nucleotides \(Y_i\) is uniform given a uniform input distribution and symmetric substitution probabilities. Thus, the entropy of the output process has a closed-form expression \(H\left(P^{\mathbf{Y}}\right) = 2\). However, for run-length \(m\)-constrained coding, \(H\left(P^{\mathbf{Y}}\right)\) has no closed-form expression, and we estimate \(H\left(P^{\mathbf{Y}}\right)\) as described in the next Subsection~\ref{MCMC}.

\subsection{MCMC Simulation-Based Estimation}
\label{MCMC}

We estimate the entropy \(H(P^{\mathbf{Y}})\) of the output process using Markov Chain Monte Carlo (MCMC) simulations, as described in~\citet{jurgens_shannon_2021}. 

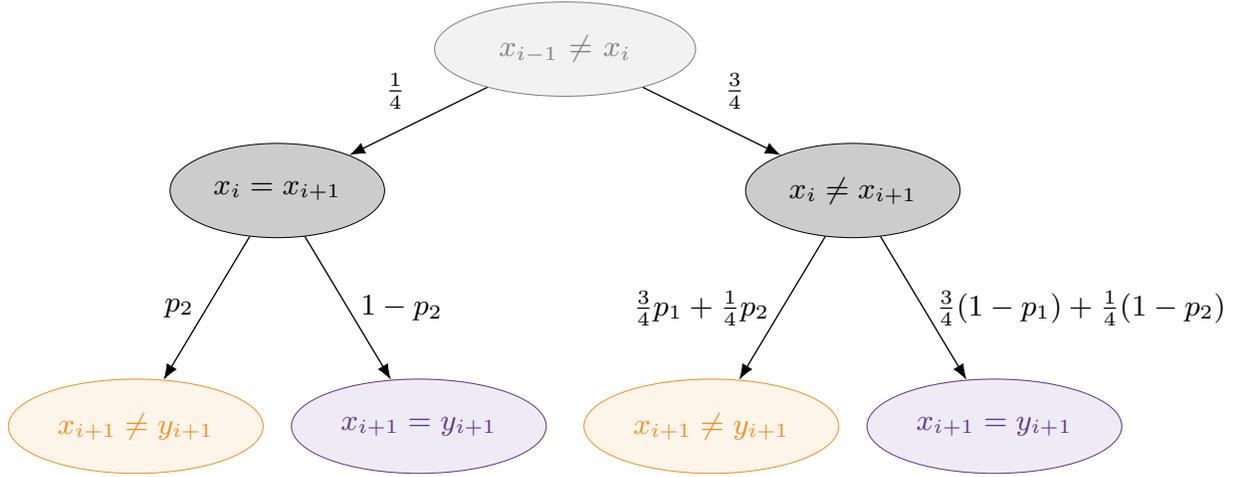
\begin{figure}[t]
    \centering
    \resizebox{\linewidth}{!}{%
    \begin{tikzpicture}[
        rootnode/.style = {ellipse, draw=gray, align=center, text width=1.7cm, minimum height=1cm, very thin, fill=gray!10, text=gray},
        greynode/.style = {ellipse, draw=black, align=center, minimum height=1cm, very thin, fill=gray!40, text=black},
        nonerrornode/.style = {ellipse, draw=Orchid, align=center, minimum height=1cm, very thin, fill=Orchid!10, text=Orchid,opacity=1},
        errornode/.style = {ellipse, draw=Peach, align=center, minimum height=1cm, very thin, fill=Peach!10, text=Peach,opacity=1},
        edge from parent/.style={draw, -Latex},
        level 1/.style={sibling distance=6.1cm, level distance=1.5cm},
        level 2/.style={sibling distance=3cm, level distance=2.5cm},
        every node/.style = {font=\footnotesize}
    ]

    \node[rootnode] (root1) {\(x_{i-1} \neq x_{i}\)}
        child {node[greynode] (left1) { \(x_{i} = x_{i+1}\)}
            child {node[errornode] (leftchild11) {\(x_{i+1} \neq y_{i+1}\)} edge from parent node[left] {$p_2$}}
            child {node[nonerrornode] (leftchild12) {\(x_{i+1} = y_{i+1}\)} edge from parent node[right] {$1-p_2$}}
            edge from parent node[above left] {\(\frac{1}{4}\)} 
        }
        child {node[greynode] (right1) { \(x_{i} \neq x_{i+1}\)}
            child {node[errornode] (rightchild11) {\(x_{i+1} \neq y_{i+1}\)} edge from parent node[left] {$\frac{3}{4}p_1+\frac{1}{4}p_2$}}
            child {node[nonerrornode] (rightchild12) {\(x_{i+1} = y_{i+1}\)} edge from parent node[right] {$\frac{3}{4}(1-p_1)+\frac{1}{4}(1-p_2)$}}
            edge from parent node[above right] {\(\frac{3}{4}\)} 
        };

    \end{tikzpicture} %
    }
    \captionsetup{labelfont=bf}
    \caption{Transition tree for a Hidden Markov Model with maximum run-length constraint \(m=2\).}
    \label{example}
\end{figure}

The process \(P^{\mathbf{Y}}\) that generates output sequences \(\mathbf{Y}\) follows a hidden Markov model. This model generates the output sequence \(\mathbf{Y}\) by transmitting a codeword \(\mathbf{X}\) through a run-length varying channel with substitution probabilities \(p_r\). The codeword \(\mathbf{X}\) is generated according to a Markovian distribution over the input nucleotides. The states of the hidden Markov model represent the transmitted nucleotides (the number of which depends on constraint \(m\)) and the current observed nucleotide. The states are hidden because the transmitted nucleotides are not directly observed due to channel noise. For example, for the constraint \(m=2\), the state space consists of \(4^3\) states, where each state is a combination of the previous transmitted nucleotide \(X_{i-1}\), the current transmitted nucleotide \(X_i\), and the current observed nucleotide \(Y_i\). The transition matrix \(\mathbf{B}\) is constructed as follows:

\begin{enumerate}
    \item For states where \(X_{i-1} = X_i\):
    \begin{itemize}
        \item There are no entries in the transition matrix for states where \(X_i = X_{i+1}\), because the next nucleotide \(X_{i+1}\) must be different from \(X_i\) to satisfy run-length constraint \(m=2\). Therefore, the probability of transitioning to these states is zero.
        \item For all remaining states where \(X_i \neq X_{i+1}\), the probability of transitioning to these states depends on the probability of the next input nucleotide \(X_{i+1}\) and the substitution probability for \(X_{i+1}\). The probability for each next nucleotide \(X_{i+1} \neq X_i\) is \(\frac{1}{3}\). The substitution probability for \(X_{i+1}\) depends on its run-length. Since \(X_i \neq X_{i+1}\), \(X_{i+1}\) does not occur in a run with \(X_i\). However, \(X_{i+1}\) can occur in a run of length $2$ with \(X_{i+2}\). Therefore, we weigh the substitution probability by the probability that \(X_{i+1}\) occurs in a run of length one (and thus is substituted with probability \(p_1\)), and the probability that it occurs in a run of length two with \(X_{i+2}\) (and thus is substituted with probability \(p_2\)).
    \end{itemize}
    \item For states where \(X_{i-1} \neq X_i\):
    \begin{itemize}
        \item The transition probabilities between states are constructed similarly, but now \(X_{i+1}\) can be one of all four nucleotides, each with probability \(\frac{1}{4}\). The transition probabilities must also account for the run-length in which \(X_{i+1}\) occurs. For states where \(X_i = X_{i+1}\), the substitution probability is \(p_2\). For states where \(X_i \neq X_{i+1}\), the substitution probability is again weighted by the probability of \(X_{i+1}\) occurring in a run of length one and a run of length two with \(X_{i+2}\).
    \end{itemize}
\end{enumerate}

Figure~\ref{example} illustrates the transition probabilities for the hidden Markov model when \(m=2\) and Figure~\ref{transitions} illustrates the transition matrices \(\mathbf{B}\) for maximum run-length constraints \(m=1, 2,\) and \(3\).

\begin{figure}[t]
  \centering
  \begin{overpic}[scale=0.9,unit=1mm]{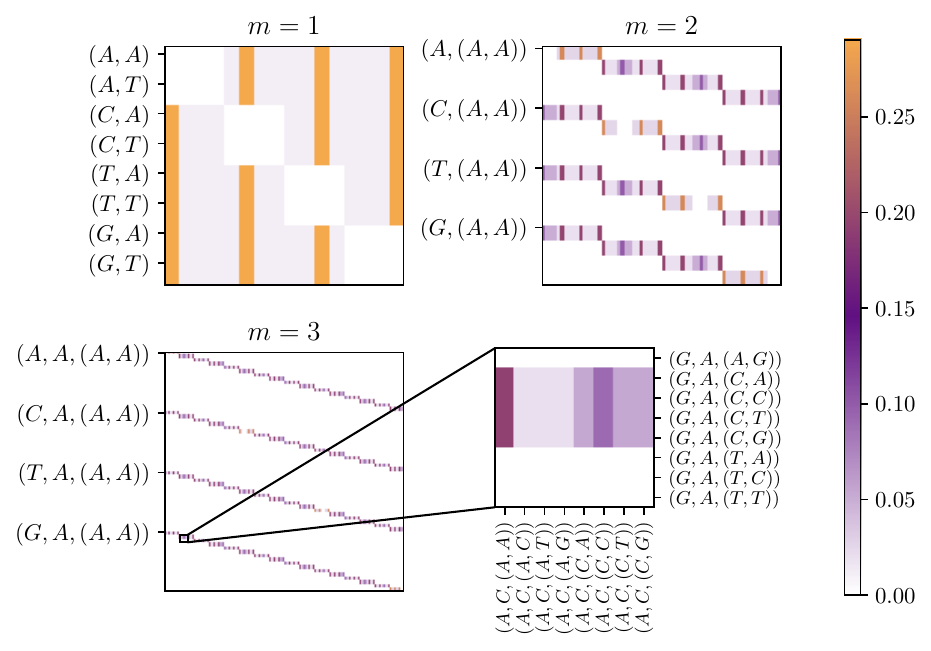}
    \put(-2,30){\rotatebox{90}{Previous State}}
    \put(25,2){Next State}   
  \end{overpic}
  \captionsetup{labelfont=bf}
  \caption{Transition matrices \( \mathbf{B} \) for maximum run-length constraints \( m=1, 2,\) and \( 3 \) with base error probability \( p_1=0.1 \) and growth factor \( \alpha=0.5 \).}
  \label{transitions}
\end{figure}

In a hidden Markov model \(P^{\mathbf{Z}}\), where the transition to the next state is uniquely determined given the current state and the observed symbol, the entropy rate has a closed-form expression and can be calculated as follows:
\[
H(P^{\mathbf{Z}}) = -\sum_{ij} \boldsymbol{\pi}_i \mathbf{B}_{ij} \log \mathbf{B}_{ij},
\]
where \(\mathbf{B}\) is the transition matrix and \(\boldsymbol{\pi}\) is the stationary distribution over the states, i.e., we can calculate the entropy by summing the entropies of individual transitions between states, weighted by how frequently we are in each state.

However, for the hidden Markov model \(P^{\mathbf{Y}}\) that we consider, the transition to the next state, given the current state and observed symbol, is non-deterministic due to the channel noise. \citet{jurgens_shannon_2021} propose restoring this deterministic property to estimate the entropy \(H(P^{\mathbf{Y}})\). The authors propose changing the state representation of the hidden Markov model to use mixed states, denoted by the vector \(\boldsymbol{\eta}\). The mixed state vector \(\boldsymbol{\eta}_i\) captures the uncertainty about which state the process is in after observing a sequence of nucleotides \(Y_1 Y_2 \cdots Y_{i-1}\) and can be interpreted as the decoder's belief at time \(i\).

The mixed state \(\boldsymbol{\eta}_0\) at time \(0\), before any output nucleotides are observed, is initialized with the stationary distribution over the states in transition matrix \(\mathbf{B}\). After observing symbol \(Y_i\), the decoder updates its belief about being in each state (i.e., updates its mixed state vector) according to:
\begin{align}
\boldsymbol{\eta}_{i+1} &= \frac{\boldsymbol{\eta}_i \mathbf{B}^{Y_i}}{\boldsymbol{\eta}_i \mathbf{B}^{Y_i} \mathbf{1}}, 
\label{update}
\end{align}
where \(\mathbf{1}\) is a vector of all ones and \(\mathbf{B}^{Y_i}\) is the nucleotide-specific transition matrix of the hidden Markov model in its original state representation. \(\mathbf{B}^{Y_i}\) is constructed by setting all entries in the transition matrix \(\mathbf{B}\) to zero where the current observed symbol is not \(Y_i\).

The probability of observing the next nucleotide \(Y_{i+1} \in \{A, C, G, T\}\) can then be computed by weighting the nucleotide-specific transition matrix \(\mathbf{B}^{Y_{i+1}}\) according to the mixed state vector \(\boldsymbol{\eta}_i\):
\[
\Pr(Y_{i+1} \mid \boldsymbol{\eta}_i) = \boldsymbol{\eta}_i \mathbf{B}^{Y_{i+1}} \mathbf{1}.
\]

With this new state representation, the hidden Markov model is now deterministic. Given the current mixed state \(\boldsymbol{\eta}_i\) and the observed symbol \(Y_i\), the transition to the next mixed state vector \(\boldsymbol{\eta}_{i+1}\) is uniquely determined by Equation~\eqref{update}. 

However, a challenge with this mixed state representation is that the set of mixed states is typically infinite. To address this, \citet{jurgens_shannon_2021} propose estimating the entropy of hidden Markov models in the mixed state representation by analyzing a finite, but sufficiently long, trajectory of mixed state vectors:  
\begin{equation*}
\hat{H}(P^{\mathbf{Y}}) = -\lim_{n \to \infty} \frac{1}{n} \sum_{i=0}^{n} \sum_{Y_{i+1} \in \{A, C, G, T\}} \Pr(Y_{i+1} \mid \boldsymbol{\eta}_i) \log_2 \left( \Pr(Y_{i+1} \mid \boldsymbol{\eta}_i) \right)
\end{equation*}

\citet{jurgens_shannon_2021} show that for hidden Markov models with transition matrix \(\mathbf{B}\) nonnegative, irreducible, and aperiodic, we have: 
\[
\hat{H}(P^{\mathbf{Y}}) \rightarrow H(P^{\mathbf{Y}}) \quad \text{as} \quad n \to \infty.
\]

Algorithm~\ref{algorithm} provides the pseudocode that summarizes how we obtain the entropy estimate \(\hat{H}(P^{\mathbf{Y}})\) using \citet{jurgens_shannon_2021}'s method.

\begin{algorithm}[H]
\caption{Entropy Convergence of a Hidden Markov Process in Mixed State Representation}
\label{algorithm}

\begin{algorithmic}[1]
\Require{$\text{convThresh}$, $\text{stabReq}$}
\Initialize{$Y_i \gets \{A,C,G,T\}$, $\boldsymbol{\eta} \gets \boldsymbol{\pi}$, entropy $\gets 0$, $H(P^{\mathbf{Y}}) \gets [\,]$, $\hat{H}(P^{\mathbf{Y}}) \gets [\,]$, $\text{stabCount} \gets 0$, $\text{converged} \gets \text{False}$}
\For{each $y_i$ in $Y_i$}
  \State $\mathbf{B}^{y_i} \gets \Call{generate\_symbol\_transition\_matrix}{y_i}$
  \State $\mathbf{B} \gets \mathbf{B} + \mathbf{B}^{y_i}$ 
\EndFor
\State Check if $\mathbf{B}$ is nonnegative, irreducible, and aperiodic
\State $\boldsymbol{\eta} \gets \Call{calculate\_stationary\_distribution}{\mathbf{B}}$
\While{not $\text{converged}$}
    \For{each $y_i$ in $Y_i$}
    \State $\Pr(y_i \mid \boldsymbol{\eta}) \gets \boldsymbol{\eta} \mathbf{B}^{y_i} \mathbf{1}$
    \If{$\Pr(y_i \mid \boldsymbol{\eta}) > 0$}
        \State $\text{entropy} \gets \text{entropy} + (-\Pr(y_i \mid \boldsymbol{\eta}) \cdot \log_2(\Pr(y_i \mid \boldsymbol{\eta})))$
    \EndIf
    \State Append $\Pr(y_i \mid \boldsymbol{\eta})$ to probabilities
\EndFor
    \State $y_i \gets$ Sample from $Y_i$ based on probabilities
    \State $\boldsymbol{\eta} \gets \frac{\boldsymbol{\eta} \mathbf{B}^{y_i}}{\boldsymbol{\eta} \mathbf{B}^{y_i} \mathbf{1}}$
    \State Append $\text{entropy}$ to $H(P^{\mathbf{Y}})$
    \State $\hat{H}(P^{\mathbf{Y}}) \gets$ Mean of $H(P^{\mathbf{Y}})$
    \If{$|\hat{H}(P^{\mathbf{Y}}) - \text{prev}\hat{H}(P^{\mathbf{Y}})| < \text{convThresh}$}
        \State $\text{stabCount} \gets \text{stabCount} + 1$
    \Else
        \State $\text{stabCount} \gets 0$
    \EndIf
    \If{$\text{stabCount} = \text{stabReq}$}
        \State $\text{converged} \gets \text{True}$
    \EndIf
    \State $\text{prev}\hat{H}(P^{\mathbf{Y}}) \gets \hat{H}(P^{\mathbf{Y}})$
\EndWhile
\end{algorithmic}
\end{algorithm}

\newpage 
\section{Proof of Lemma~\ref{HB}}
\label{pf2}

Assume, without loss of generality, that the sequence length \( n \) is even. Let us first calculate the Hamming ball volume when \(\epsilon = 0\) and then generalize to \(0 \leq \epsilon \leq 0.5\). We apply the expression from \citet{king_bounds_2004} for constant GC contents $w$ by setting \(w = 0.5n\). This gives the Hamming ball volume for any sequence $\mathbf{X}$ in constrained space $\mathcal{S}_{0}$, where $\mathcal{S}_{0}$ is the subset of sequences $\mathbf{X} \in \{A,C,G,T\}^n$ with balanced GC content $w=0.5n$:
    \begin{equation}
        V_{0}(\mathbf{X})=\sum_{r=0}^{d-1}\sum_{i=0}^{\min\left(\left\lfloor\frac{r}{2}\right\rfloor, 0.5n\right)} \binom{0.5n}{i}\binom{n-0.5n}{i}\binom{n-2i}{r-2i}2^{2i}.
        \label{GV_0}
    \end{equation}

The outer summation iterates over all possible Hamming distances up to $d-1$. The inner summation accounts for all possibilities in which $i$ substitutions of nucleotides are made from $\{G, C\}$ to $\{A, T\}$ and vice versa, ensuring each substitution is counterbalanced to maintain a balanced GC content of $w=0.5n$. This results in a total of $2i$ substitutions. The binomial coefficients $\binom{0.5n}{i}$ and $\binom{n-0.5n}{i}$ calculate the number of ways to select $i$ nucleotides for substitution from the $0.5n$ nucleotides within $\{G, C\}$, and from the $n-0.5n$ nucleotides within $\{A, T\}$, respectively. The term $2^{2i}$ counts the number of possibilities for these $2i$ substitutions, considering that each substituted nucleotide can be replaced with either of two options (either $A$ or $T$ for a nucleotide from $\{G, C\}$ and vice versa). Finally, $\binom{n-2i}{r-2i}$ accounts for the remaining substitutions needed to achieve a Hamming distance exactly $r$ from the center sequence. Since the remaining substitutions must also preserve the GC content, there is only one possible substitution for each selected position (i.e., if a position with $G$ is selected, it can only be substituted to $C$ and vice versa, and similarly for substitutions within the $\{A, T\}$ group).

We extend the result by \citet{king_bounds_2004} to constrained spaces \( \mathcal{S}_{\epsilon} \) by accounting for the possibility that sequences within the Hamming ball $|V_{\epsilon}(\mathbf{X})|$ can differ in GC content to their center sequence $\mathbf{X}$ (constrained spaces \( \mathcal{S}_{\epsilon} \) allow for a range of GC contents determined by $\epsilon$).

We classify substitutions into three categories based on how they affect the sequence's GC content:

\begin{itemize}
    \item \textbf{Increasing substitutions} (\(i_+\)): Substitutions that change nucleotides from the set \(\{A, T\}\) to those in the set \(\{G, C\}\), thereby increasing the GC content of the sequence.
    \item \textbf{Decreasing substitutions}: Substitutions that change nucleotides from the set \(\{G, C\}\) to those in the set \(\{A, T\}\), thereby reducing the GC content of the sequence.
    \item \textbf{Preserving substitutions}: Substitutions within a single nucleotide group (i.e., from \(G\) to \(C\) and vice versa, and from \(A\) to \(T\) and vice versa) that do not alter the GC content of the sequence.
\end{itemize}

We calculate the volume of the Hamming ball \( |V_{\epsilon}(\mathbf{X})| \) centered at a sequence $\mathbf{X}$ with GC content $w$ by considering all substitution combinations that keep the GC content within a deviation $\epsilon$ from a balanced GC content, i.e., within constrained space \( \mathcal{S}_{\epsilon} \):

\begin{align}
|V_{\epsilon}(\mathbf{X})| = \sum_{r=0}^{d-1} \sum_{\Delta = \max(\lceil(0.5 - \epsilon) n\rceil - w, -r)}^{\min(\lfloor(0.5 + \epsilon) n\rfloor - w, r)}  \sum_{i_+ = \max(0, \Delta)}^{\min(\Delta + w, r)}  \binom{w}{i_+ - \Delta} \binom{n-w}{i_+} \binom{n-i_+ - (i_+ - \Delta)}{r-i_+ - (i_+ - \Delta)} 2^{i_+ + (i_+ - \Delta)}. 
\end{align}

The middle summation iterates over \(\Delta = i_+ - i_-\), which is the total change in GC content due to substitutions. The range of \(\Delta\) is constrained by:
\begin{itemize}
    \item \textbf{Maximum Decrease:} A negative \(\Delta\) indicates a total decrease in GC content. The largest allowable decrease is constrained by the smaller of two values: the total number \(r\) of substitutions, and the maximum decrease allowed by constraint \(\epsilon\), taking into account the GC content $w$ of the sequence. Thus, the lowest value for \(\Delta\) is \(\max(\lceil (0.5 - \epsilon) n \rceil - w, -r)\).
    \item \textbf{Maximum Increase:} Conversely, a positive \(\Delta\) indicates a total increase in GC content. The upper limit for \(\Delta\) is determined by the smaller of two values: \(r\) and the difference between the maximum permissible GC content defined by constraint \(\epsilon\) and the GC content \(w\) of the sequence. Thus, the highest value for \(\Delta\) is \(\min(\lfloor (0.5 + \epsilon) n \rfloor - w, r)\).
\end{itemize}

The inner summation iterates over all possible combinations of \(i_+\) and \(i_-\) that achieve total GC content change \(\Delta\). The binomial coefficient \(\binom{n}{k}\) equals zero when \(k > n\). This ensures that we do not make more increasing or decreasing substitutions than the available \(\{A,T\}\) and \(\{G,C\}\) nucleotides in the sequence.




\section{Supplementary results: Homopolymers}
\label{App C}

In Section~\ref{lbr}, we consider a linear growth model for the substitution rate increase. In this supplement, we consider two additional growth models. 



Figure~\ref{App10} shows the error regimes in which $m$-constrained coding is more efficient than unconstrained coding and vice versa for an exponential growth model $p_{r}=\min (0.75, \hspace{0.1cm} pe^{\alpha(r-1)})$ with $p=1\%$. 
The relative differences in substitution rate increase slowly for shorter runs and quickly for longer runs. At the same time, the probability of reading a long run goes to zero. Therefore, the increase in substitution rate at which $m$-constrained coding becomes efficient is larger in the exponential growth model than in the linear one. The findings are consistent in that the increase in substitution rate must be large for $m$-constrained coding to be efficient.

\begin{figure}[t]
    \centering
    \includegraphics[scale=1]{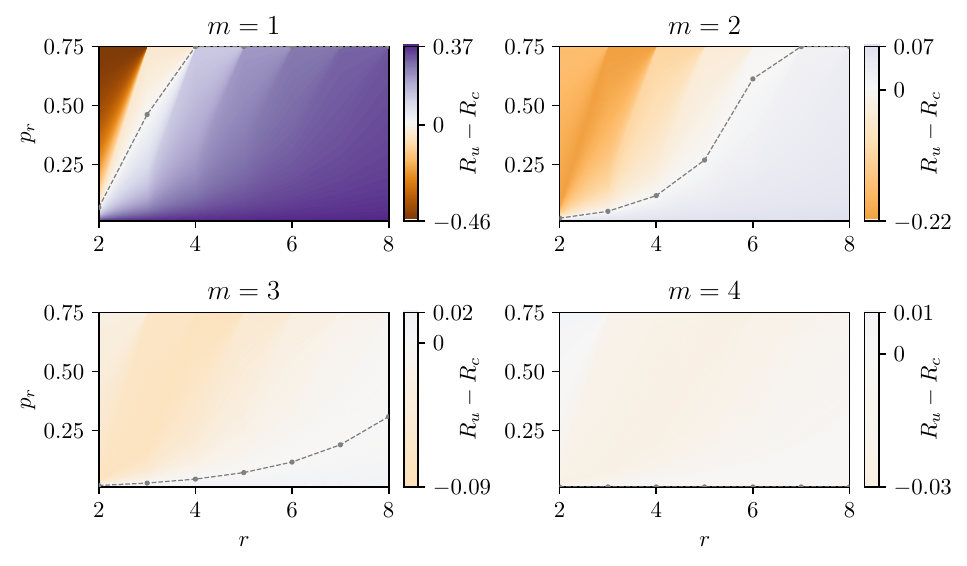}
\captionsetup{labelfont=bf}
    \caption{Error regimes for an exponential growth model $p_{r}=\min (0.75, \hspace{0.1cm} pe^{\alpha(r-1)})$.}
    \label{App10}
\end{figure}

Figure~\ref{App9} shows the error regimes in which $m$-constrained and unconstrained are more efficient for a logarithmic growth model $ p_{r}= \min (0.75, \hspace{0.1cm} \alpha \ln(r)+p)$ with $p=1\%$. The increase in substitution rate at which $m$-constrained becomes efficient is smaller in the logarithmic growth model than in the linear one. However, the results are consistent in that a large increase in the substitution rate is necessary for constraint $m=1$ to be efficient. For weaker constraints $m=2, 3$ and $4$, the difference in code rate goes to zero and unconstrained coding is always preferred due to its less complex code design. 

\begin{figure}[t]
    \centering
    \includegraphics[scale=1]{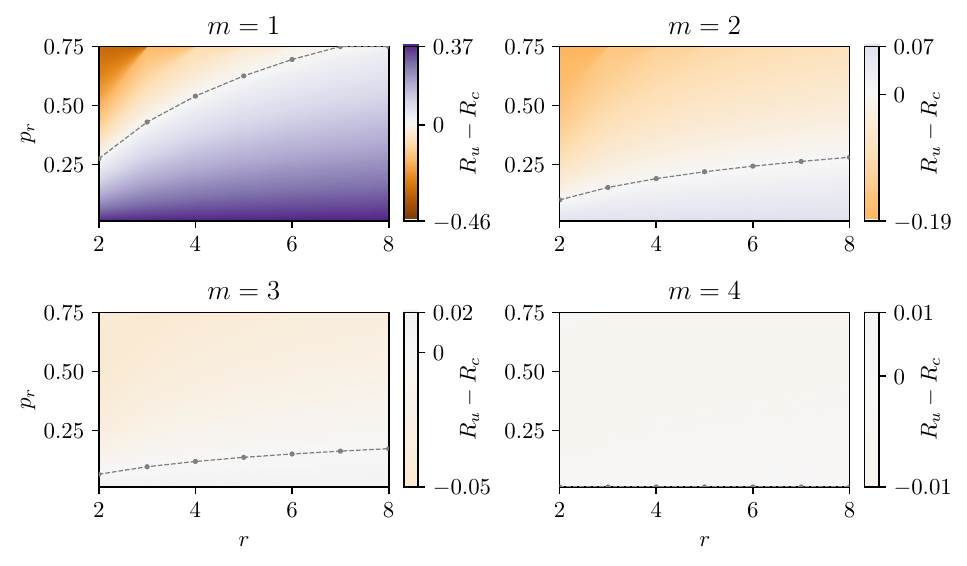}
\captionsetup{labelfont=bf}
   \caption{Error regimes for a logarithmic growth model $ p_{r}= \min (0.75, \hspace{0.1cm} \alpha \ln(r)+p)$.}
    \label{App9}
\end{figure}

\section{Supplementary results: Empirical error analysis}
\label{App D}

In this section, we provide additional details on the empirical results.

\paragraph{Run-length distribution.}
We start by providing an overview of the descriptive statistics for the datasets considered in Section~\ref{emp} to explain in more detail why we restrict our analysis to runs up to and including length six and GC content between $35-65\%$. 

Figure~\ref{App12} shows the frequency of the run-lengths in the datasets considered. As expected, long runs are rare in all datasets considered. The sample size of nucleotides that are part of runs longer than six is too small to obtain reliable error rate estimates. 

\begin{figure}[t]
    \centering
    \begin{tikzpicture}
        \begin{groupplot}[
            group style={
                group size=4 by 2,
                vertical sep=1.9cm,
                horizontal sep=1cm,
                ylabels at=edge left,  
                xlabels at=edge bottom 
            },
            ybar,
            xmax=10,
            xlabel=run-length $r$,
            ylabel=Frequency,
            width=4.5cm,
            height=3.8cm,
            grid=none, 
            xtick style={draw=none},
            title style={yshift=0.9ex}, 
        ]
            
            \nextgroupplot[title={Srinivasavaradhan et al.}]
            \addplot+[bar width=6.5pt, fill=Orchid!40, draw=Orchid] table[x=length, y=n_seqs] {Srinivasavaradhan_sub.dat};

            \nextgroupplot[title={Antkowiak et al.}, xmax=10, xtick={2,4,6,8,10,12}]
            \addplot+[bar width=6.5pt, fill=Orchid!40, draw=Orchid] table[x=length, y=n_seqs] {Antkowiak_sub.dat};

            \nextgroupplot[title={Netflix-Pool}, xmax=10, xtick={2,4,6,8,10,12}]
            \addplot+[bar width=6.5pt, fill=Orchid!40, draw=Orchid] table[x=length, y=n_seqs] {netflix_del.dat};
            
            \nextgroupplot[title={Gimpel et al. \uppercase\expandafter{\romannumeral1}}]
            \addplot+[bar width=6.5pt, fill=Orchid!40, draw=Orchid] table[x=length, y=n_seqs] {GGall_sub.dat};
            
            \nextgroupplot[title={Gimpel et al. \uppercase\expandafter{\romannumeral2}}]
            \addplot+[bar width=6.5pt, fill=Orchid!40, draw=Orchid] table[x=length, y=n_seqs] {GGfix_sub.dat};

            \nextgroupplot[title={Gimpel et al. \uppercase\expandafter{\romannumeral3}}]
            \addplot+[bar width=6.5pt, fill=Orchid!40, draw=Orchid] table[x=length, y=n_seqs] {GTall_sub.dat};

            \nextgroupplot[title={Gimpel et al. \uppercase\expandafter{\romannumeral4}}]
            \addplot+[bar width=6.5pt, fill=Orchid!40, draw=Orchid] table[x=length, y=n_seqs] {GTfix_sub.dat};
        \end{groupplot}
    \end{tikzpicture}
    \captionsetup{labelfont=bf}
    \caption{
    Frequency of run-lengths in the DNA storage systems considered. Runs of lengths one and two occur with highest frequency, and long runs are observed with low frequency. This is consistent with the theoretical run distribution and its mean run-length of $1.6$ established in Section~\ref{cr}. The sample size for runs longer than six is too small to obtain reliable error rate estimates. Therefore, we limit our analysis in Section~\ref{emp} to runs up to and including length six. }
    \label{App12}
\end{figure}
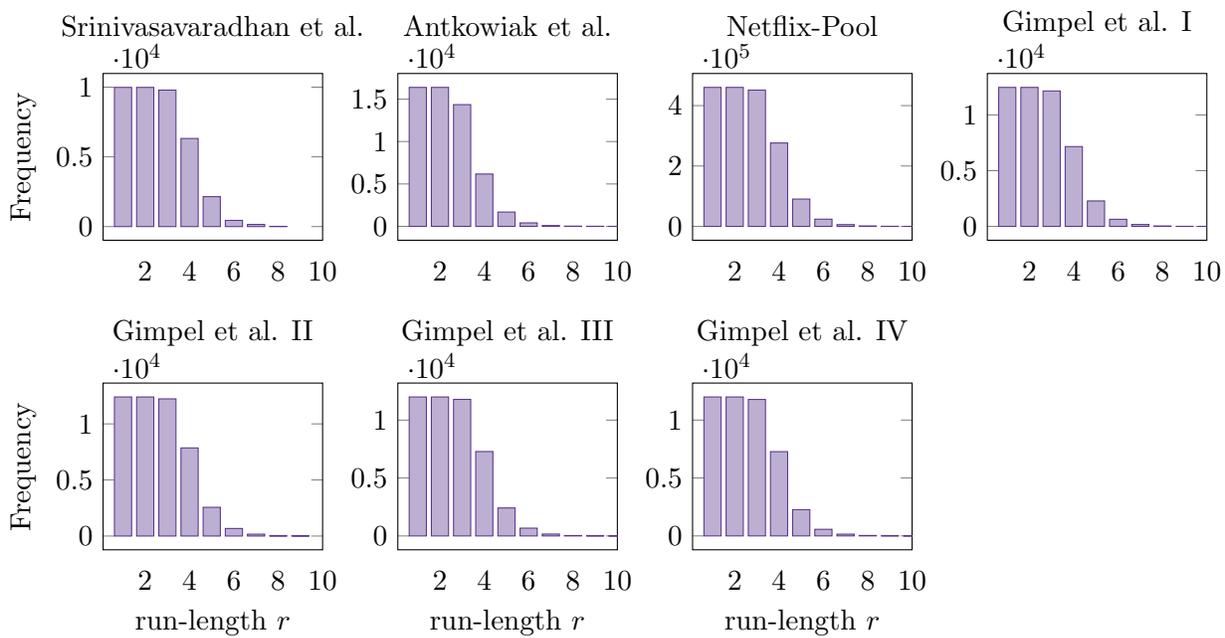

\paragraph{GC content distribution.} Similarly, Figure~\ref{App13} shows the distribution of GC content in the DNA storage systems considered. The empirical GC content distributions are consistent with the theoretical distribution established in Section~\ref{crGC}; that is, extreme unbalances in GC content are observed infrequently. Consequently, our error analysis is limited to sequences with GC content between $0.35$ and $0.65$ to ensure sufficient sample sizes for our error rate estimates. 
\begin{figure}[t]
    \centering
    \begin{tikzpicture}
        \begin{groupplot}[
            group style={
                group size=3 by 2,
                vertical sep=1.9cm,
                horizontal sep=1.1cm,
                ylabels at=edge left,  
                xlabels at=edge bottom 
            },
            ybar,
            xlabel=GC content $w$,
            ylabel=Frequency,
            width=6cm,
            height=3.8cm,
            grid=none, 
            xtick pos=bottom,
            title style={yshift=0.9ex},
            xtick style={draw=none}, 
        ]
                    
            \nextgroupplot[title={Srinivasavaradhan et al.}]
            \addplot+[bar width=2pt, fill=Orchid!40, draw=Orchid] table[x=GC, y=n_reads] {groupedM.dat};

            \nextgroupplot[title={Antkowiak et al.}]
            \addplot+[bar width=5.6pt, fill=Orchid!40, draw=Orchid] table[x=GC, y=n_reads] {groupedH.dat};

            \nextgroupplot[title={Netflix-Pool}]
            \addplot+[bar width=4.4pt, fill=Orchid!40, draw=Orchid] table[x=GC, y=n_reads] {grouped.dat};

            \nextgroupplot[title={Gimpel et al. \uppercase\expandafter{\romannumeral1}}]
            \addplot+[bar width=5.4pt, fill=Orchid!40, draw=Orchid] table[x=GC, y=n_reads] {groupedG.dat};
            
            \nextgroupplot[title={Gimpel et al. \uppercase\expandafter{\romannumeral3}}]
            \addplot+[bar width=2pt, fill=Orchid!40, draw=Orchid] table[x=GC, y=n_reads] {groupedT.dat};
        \end{groupplot}
    \end{tikzpicture}
    \captionsetup{labelfont=bf}
    \caption{Distribution of the GC content in the DNA storage systems considered. Sequences with balanced GC content occur with highest frequency, and sequences with extremely low and high GC content are observed with low frequency. This is consistent with the theoretical GC content distribution and its mean of $50\%$ established in Section~\ref{crGC}. In particular, the sample size for sequences with GC content larger than $65\%$ or smaller than $35\%$ does not provide a reliable error rate estimates. Therefore, in Section~\ref{emp}, we limit our analysis to GC content between $35\%$ and $65\%$.}
    \label{App13}
\end{figure}
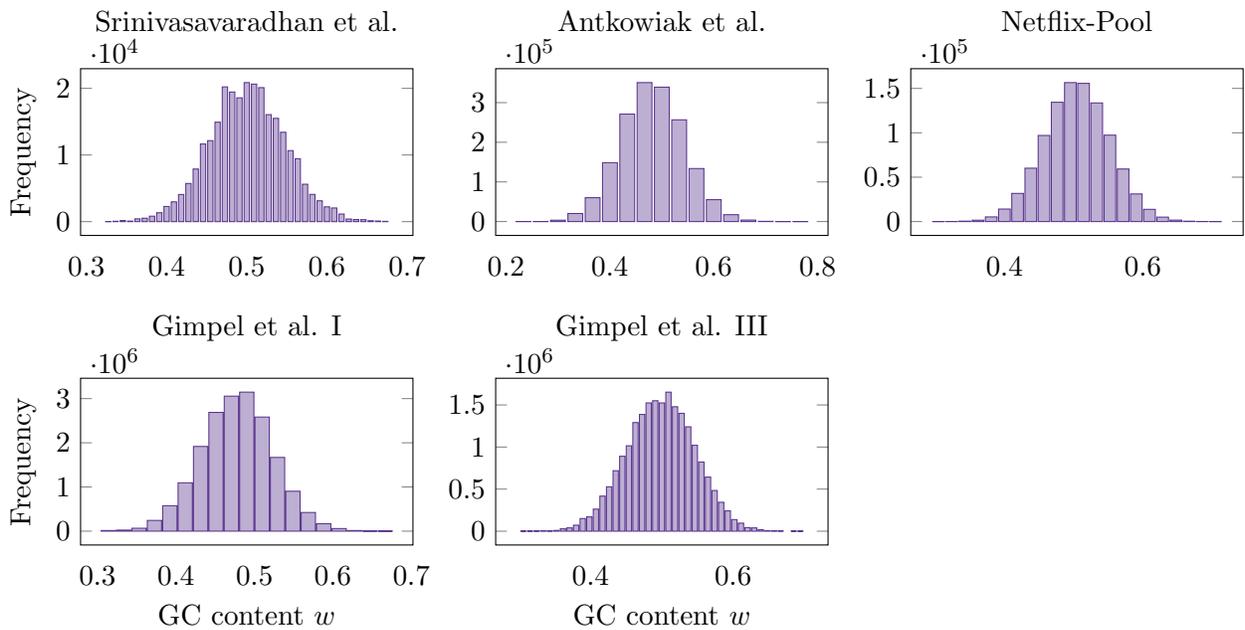

\paragraph{Read coverage.}
Empirical studies suggest that unbalances in GC content can result in non-uniform read distributions~\citep{benjamini_summarizing_2012, dohm_substantial_2008}. While our theoretical results do not address whether constrained coding is efficient in minimizing sequencing efforts, we empirically investigate possible read biases for the DNA storage systems considered.

Figure~\ref{App14} shows a correlation between GC content and read frequency in the DNA storage systems considered. However, to determine the efficiency of constrained coding in reducing sequencing efforts, the additional sequencing cost (due to unbalances in GC content) must be weighed against the additional synthesis cost (due to GC content constraints).

\begin{figure}[t]
    \centering
\begin{tikzpicture}
    \begin{groupplot}[
        group style={
                group size=3 by 2,
                vertical sep=1.7cm,
                horizontal sep=1cm,
                ylabels at=edge left,  
        xlabels at=edge bottom  
        },
        ybar,
        ylabel=rel. abundance,
        xlabel=GC content $w$,
            width=6cm,
            height=3.8cm,
        xmin=0.35,
        xmax=0.65
    ]
        
        \nextgroupplot[title={Srinivasavaradhan et al.}]
        \addplot [
    only marks,
    opacity=0.8,
    color=Peach,
    mark size=1.5,
    error bars/.cd,
    y dir=both,
    y explicit
] table [
    x=GC,
    y=read,
    y error=std_s
] {unequal_M.dat};

        \nextgroupplot[title={Antkowiak et al.}]
\addplot [
    only marks,
    opacity=0.8,
    color=Peach,
    mark size=1.5,
    error bars/.cd,
    y dir=both,
    y explicit
] table [
    x=GC,
    y=read,
    y error=std_s
] {unequal_A.dat};

\nextgroupplot[title={Netflix-Pool}]
\addplot [
    only marks,
    opacity=0.8,
    color=Peach,
    mark size=1.5,
    error bars/.cd,
    y dir=both,
    y explicit
] table [
    x=GC,
    y=read,
    y error=std_s
] {unequal_N.dat};
        
        \nextgroupplot[title={Gimpel et al. \uppercase\expandafter{\romannumeral1}}]
\addplot [
    only marks,
    opacity=0.8,
    color=Peach,
    mark size=1.5,
    error bars/.cd,
    y dir=both,
    y explicit
] table [
    x=GC,
    y=read,
    y error=std_s
] {unequal_GG.dat};
        
        \nextgroupplot[title={Gimpel et al. \uppercase\expandafter{\romannumeral3}}]
  \addplot [
    only marks,
    opacity=0.8,
    color=Peach,
    mark size=1.5,
    error bars/.cd,
    y dir=both,
    y explicit
] table [
    x=GC,
    y=read,
    y error=std_s
] {unequal_GT.dat};
    \end{groupplot}
\end{tikzpicture}
    \captionsetup{labelfont=bf}
    \caption{Sequencing bias due to GC content unbalances. In all datasets, sequences with low or high GC content are read fewer times than sequences with balanced GC content. The observed trend may be because sequences with extreme GC content are amplified less efficiently during PCR amplification \citep{kozarewa_amplification-free_2009}.}
    \label{App14}
\end{figure}
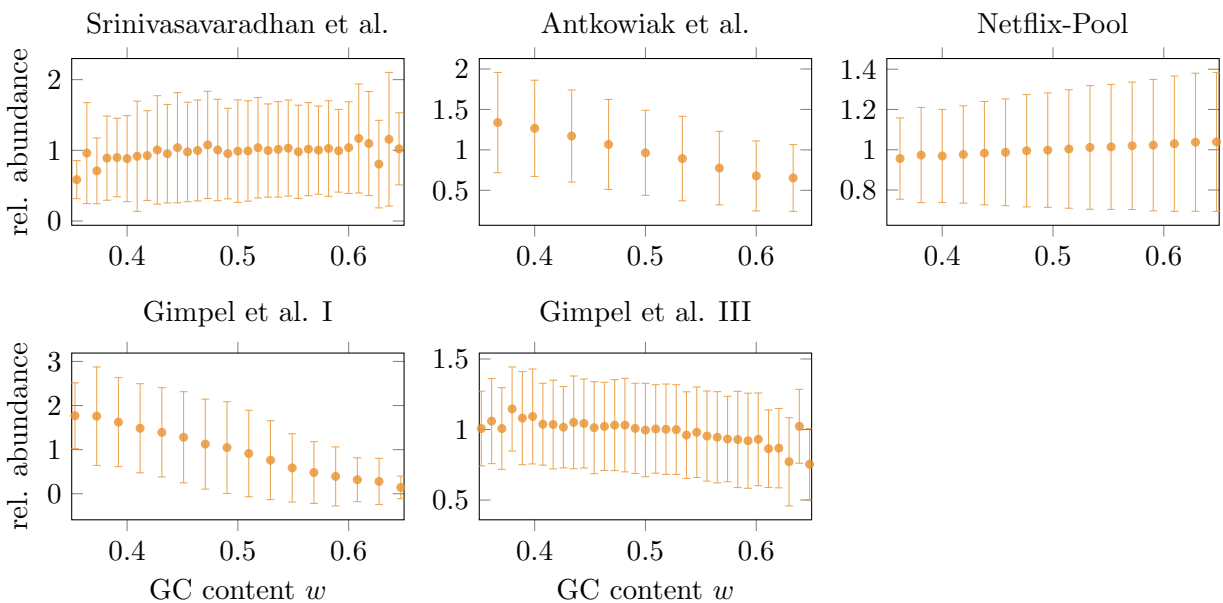


\end{document}